\documentclass[11pt]{article}
\usepackage{amsmath,stackrel}
\usepackage{amsthm}
\allowdisplaybreaks
\usepackage{amssymb}
\usepackage{algorithm}
\usepackage[lined,boxed,ruled,norelsize,algo2e,linesnumbered]{algorithm2e}
\usepackage{subfig}

\usepackage{array}
\usepackage{color}
\usepackage[dvipsnames]{xcolor}
\usepackage{graphicx}
\usepackage{wrapfig,epsfig}
\usepackage{epstopdf}
\usepackage{url}
\usepackage{graphicx}
\usepackage{color}
\usepackage{epstopdf}
\usepackage{algpseudocode}
\usepackage{scrextend}
\usepackage[T1]{fontenc}
\usepackage{bbm}
\usepackage{comment}
\usepackage{authblk}
\usepackage{multicol}
\usepackage{dsfont}
\usepackage{mathtools}
\usepackage{enumitem}
\usepackage{mathrsfs}
\usepackage{tikz}
\usepackage{hyperref}  
\hypersetup{colorlinks=true,citecolor=blue,linkcolor=red}

\usetikzlibrary{arrows}
\usepackage[margin=1in]{geometry}
\graphicspath{{./figs/}}

\definecolor{b2}{RGB}{51,153,255}
\definecolor{mygreen}{RGB}{80,180,0}

\theoremstyle{plain}
\newtheorem{theorem}{Theorem}[section]
\newtheorem{lemma}[theorem]{Lemma}
\newtheorem{lem}[theorem]{Lemma}

\newtheorem{proposition}[theorem]{Proposition}

\newtheorem{claim}[theorem]{Claim}
\newtheorem{corollary}[theorem]{Corollary}

\newtheorem{question}[]{Question}

\theoremstyle{definition}
\newtheorem{definition}[theorem]{Definition}

\newtheorem{problem}[theorem]{Problem}

\theoremstyle{remark}
\newtheorem{remark}[theorem]{Remark}

\newcommand{\wh}{\widehat}

\newcommand{\ov}{\overline}
\newcommand{\eps}{\varepsilon}
\renewcommand{\epsilon}{\varepsilon}
\renewcommand{\phi}{\varphi}

\newcommand{\HF}{\hat{F}}

\newcommand{\calP}{\mathcal{P}}

\renewcommand{\hat}{\wh}
\renewcommand{\bar}{\ov}
\renewcommand{\d}{\mathrm{d}}

\newcommand{\Ind}{\mathbf{1}}

\newcommand{\G}{\mathcal{G}}

\newcommand{\defeq}{\stackrel{{\text{def}}}{=}}

\DeclareMathOperator*{\E}{\mathbb{E}}
\DeclareMathOperator*{\M}{\mathcal{M}}

\DeclareMathOperator{\poly}{poly}

\DeclareMathAlphabet{\mathpzc}{OT1}{pzc}{m}{it}

\DeclarePairedDelimiterX{\xdivergence}[2]{(}{)}{%
  #1\;\delimsize\|\;#2%
}
\newcommand{\deltacurve}{\delta\xdivergence*}
\newcommand{\tradeoff}{T\xdivergence*}

\newcommand{\lp}{\left(}
\newcommand{\lb}{\left[}
\newcommand{\lc}{\left\{}

\newcommand{\rp}{\right)}
\newcommand{\rb}{\right]}
\newcommand{\rc}{\right\}}

\def\hpi{\hat{\pi}}
\def\hx{\hat{x}}
\def\hv{\hat{v}}

\def\TF{{\Tilde{F}}}

\def\barS{\bar{S}}

\newcommand{\R}{\mathbb{R}} % REALS
\newcommand{\bbP}{\mathbb P}

\newcommand{\cA}{\mathcal A}

\newcommand{\cD}{\mathcal D}

\newcommand{\cK}{\mathcal K}

\newcommand{\cM}{\mathcal M}
\newcommand{\cN}{\mathcal N}

\newcommand{\cP}{\mathcal P}

\newcommand{\cV}{\mathcal V}

\newcommand{\cY}{\mathcal Y}

\newcommand{\indicator}{\mathbbm{1}} % INDICATOR

\newcommand{\sign}{\mathrm{sign}}

\newcommand{\inpro}[2]{\left\langle #1,#2 \right\rangle} % INNER PRODUCT
\newcommand{\norm}[1]{\left\lVert #1\right\rVert} % ||X||
\renewcommand{\epsilon}{\varepsilon}
  \newcommand{\beq}{\begin{equation}}
  \newcommand{\eeq}{\end{equation}}
  \newcommand{\beqn}{\begin{equation*}}
  \newcommand{\eeqn}{\end{equation*}}
  \newcommand{\beqr}{\begin{eqnarray}}
  \newcommand{\eeqr}{\end{eqnarray}}
  \newcommand{\beqrn}{\begin{eqnarray*}}
  \newcommand{\eeqrn}{\end{eqnarray*}}
  \newcommand{\bmline}{\begin{multline}}
  \newcommand{\emline}{\end{multline}}
  \newcommand{\bmlinen}{\begin{multline*}}
  \newcommand{\emlinen}{\end{multline*}}

\title{Private Convex Optimization via Exponential Mechanism}
\author{
Sivakanth Gopi\thanks{Microsoft Research. Email: \texttt{sigopi@microsoft.com}}
\quad
Yin Tat Lee \thanks{University of Washington and Microsoft Research. Email: \texttt{yintat@uw.edu}}
\quad
Daogao Liu \thanks{University of Washington. Email: \texttt{dgliu@uw.edu}}
}
\date{}
\begin{document}

\begin{titlepage}
\maketitle
\begin{abstract}
    In this paper, we study private optimization problems for non-smooth convex functions $F(x)=\mathbb{E}_i f_i(x)$ on $\mathbb{R}^d$.
We show that modifying the exponential mechanism by adding an $\ell_2^2$ regularizer to $F(x)$ and sampling from $\pi(x)\propto \exp(-k(F(x)+\mu\|x\|_2^2/2))$ recovers both the known optimal empirical risk and population loss under $(\eps,\delta)$-DP. Furthermore, we show how to implement this mechanism using $\widetilde{O}(n \min(d, n))$ queries to $f_i(x)$ for the DP-SCO where $n$ is the number of samples/users and $d$ is the ambient dimension.
We also give a (nearly) matching lower bound $\widetilde{\Omega}(n \min(d, n))$ on the number of evaluation queries. 

Our results utilize the following tools that are of independent interest:
\begin{itemize}
    \item We prove Gaussian Differential Privacy (GDP) of the exponential mechanism if the loss function is strongly convex and the perturbation is Lipschitz. Our privacy bound is \emph{optimal} as it includes the privacy of Gaussian mechanism as a special case and is proved using the isoperimetric inequality for strongly log-concave measures.
    \item We show how to sample from $\exp(-F(x)-\mu \|x\|^2_2/2)$ for $G$-Lipschitz $F$ with $\eta$ error in total variation (TV) distance using $\widetilde{O}((G^2/\mu) \log^2(d/\eta))$ unbiased queries to $F(x)$. This is the first sampler whose query complexity has \emph{polylogarithmic dependence} on both dimension $d$ and accuracy $\eta$.
\end{itemize}

\end{abstract}
  \thispagestyle{empty}
\end{titlepage}

{\hypersetup{linkcolor=black}
\tableofcontents
}
\newpage
\section{Introduction}
Differential Privacy (DP), introduced in~\cite{DMNS06,DKMMN06}, is increasingly becoming the universally accepted standard in privacy protection. We see an increasing array of adoptions in industry~\cite{Apple17,EPK14,BEM+17,DKY17} and more recently the US census bureau \cite{Abo16,KCK+18}. 
Differential privacy allows us to quantify the privacy loss of an algorithm and is defined as follows.

\begin{definition}[$(\epsilon,\delta)$-DP]
A randomized mechanism $\M$ is $(\epsilon,\delta)$-differentially private if for any neighboring databases $\cD,\cD'$ and any subset $S$ of outputs, one has
\begin{align*}
    \Pr[\M(\cD)\in S]\leq e^\epsilon\Pr[\M(\cD')\in S] +\delta.
\end{align*} In this paper, we say $\cD$ and $\cD'$ are neighboring databases if they agree on all the user inputs except for a single user's input.
\end{definition}

Privacy concerns are particularly acute in machine learning and optimization using private user data. Suppose we want to minimize some loss function $F(x;\cD):\cK\to \R$ for some domain $\cK$ where $\cD$ is some database.
We want to output a solution $x^{priv}$ using differentially private mechanism $\cM$ such that we minimize the \emph{excess empirical risk}
\begin{equation}
\label{eqn:empirical_risk}
\E_{\cM}[F(x^{priv};\cD)]-F(x^*;\cD),    
\end{equation}
 where $x^*\in \cK$ is the true minimizer of $F(x;\cD)$. 
 
 \paragraph{Exponential Mechanism} One of the first mechanisms invented in differential privacy, the \emph{exponential mechanism}, was proposed by \cite{MT07} precisely to solve this. It involves sampling $x^{priv}$ from the density 
\begin{equation}
    \label{eqn:exponential_mechanism}
    \pi_\cD(x)\propto \exp\lp-kF(x;\cD)\rp.
\end{equation}
 Here $k$ controls the privacy-vs-utility tradeoff, large $k$ ensures that we get a good solution but less privacy and small $k$ ensures that we get good privacy but we lose utility. Suppose $\Delta_F = \sup_{\cD\sim\cD'} \sup_x |F(x;\cD)-F(x;\cD')|$ is the sensitivity of $F$, where the supremum is over all neighboring databases $\cD,\cD'$. Then choosing $k=\frac{\eps}{2\Delta_F}$, the exponential mechanism satisfies $(\eps,0)$-DP. 
 
 Exponential mechanism is widely used both in theory and in practice, such as in mechanism design \cite{HK12}, convex optimization \cite{BST14,MV21}, statistics \cite{WZ10,WM10,AKR+19}, machine learning and AI \cite{ZP19}. Even for infinite and continuous domains, exponential mechanism can be implemented efficiently for many problems \cite{HT10,CSS13,KT13,BV19,CKS20}. There are also several variants and generalizations of the exponential mechanism which can improve its utility based on different assumptions \cite{TS13,BNS13,RS16,LT19}. See \cite{LT19} for a survey of these results.

% One of the important case that exponential mechanism is known to be sub-optimal is convex optimization. Under $(\eps,\delta)$-DP, \cite{BST14,bftt19,bfgt20} shows that noisy stochastic gradient descent achieves the empirical risk roughly $\sqrt{d}/n$ for $\cK\subset \R^d$ with $n$ users. However, 
 
% Exponential mechanism can be efficiently implemented for both finite domains and infinite and continuous domains \cite{HT10,CSS13,KT13,BV19,CKS20}.
 
 %when $\cK$ is a finite.

 %Suppose $\Delta_F = \sup_{\cD\sim\cD'} \sup_x |F(x;\cD)-F(x;\cD')|$ is the sensitivity of $F$, where the supremum is over all neighboring databases $\cD,\cD'$. Then choosing $k=\frac{\eps}{2\Delta_F}$, the exponential mechanism satisfies $(\eps,0)$-DP. 
 %For convex loss $F$, the empirical risk in (\ref{eqn:empirical_risk}) is bounded by $\frac{d}{k}=O\lp\frac{d\Delta_F}{\eps}\rp$ and this is also known to be optimal under pure differential privacy~\cite{BST14}. 

% In this paper, we study on the case $F$ is convex and $\cK$ is a convex set. We show one can significantly improve the bound by adding a $\ell^2_2$ regularizer term to $F$ under $(\eps,\delta)$-DP.

% One can significantly improve the empirical risk from $O_\eps(d)$ to $O_{\eps,\delta}(\sqrt{d})$ under convexity assumptions on $F$ and by asking for $(\eps,\delta)$-DP instead of $(\eps,0)$-DP. 
%In this paper, we st
% We will now define this setting formally. 
 \paragraph{DP Empirical Risk Minimization (DP-ERM)}
 In many applications, the loss function is given by the average of the loss of each user:
   \begin{align}
\label{eq:DPERM}
    F(x;\cD):=\frac{1}{n}\sum_{i=1}^n f(x; s_i).
\end{align} 
 where $\cD = \{s_1,s_2,\cdots,s_n\}$ is the collection of users $s_i$ and $f(x; s_i)$ is the loss function of user $s_i$. 
 
 Throughout this paper, we assume $f(x;s)$ is convex and $f(x;s)-f(x;s')$ is $G$-Lipschitz for all $s,s'$, and $\cK\subset \R^d$ is convex with diameter $D$.\footnote{Some of our results can handle the unconstrained domain, such as $\cK=\R^d$.} We call the problem of minimizing the excess empirical risk in \eqref{eq:DPERM} as DP Empirical Risk Minimization (DP-ERM). This setting is well studied by the DP community with many exciting results \cite{CM08,rbht09,cms11,jt14,BST14,kj16,fts17,zzmw17,Wang18,ins+19,bftt19,FKT20,KLL21,bgn21,LL21,AFKT21,sstt21,MBST21,GTU22}.\footnote{Most of the literature uses a stronger assumption that $f(x;s)$ is $G$-Lipschitz, while some of our results only need to assume the difference $f(x;s)-f(x;s')$ is $G$-Lipschitz.}
 
In particular, \cite{BST14} shows that exponential mechanism in (\ref{eqn:exponential_mechanism}) achieves the optimal excess empirical risk of $O\lp \frac{GDd}{n\eps}\rp$ under $(\eps,0)$-DP. On the other hand, \cite{BST14,bftt19,bfgt20} show that \emph{noisy gradient descent} on $F(x;\cD)$ achieves an excess empirical risk of 
\begin{equation}
\label{eqn:optimal_empirical_risk}
O\lp \frac{GD\sqrt{d \log(1/\delta)}}{n\eps}\rp    
\end{equation}
under $(\eps,\delta)$-DP, which is also shown to be optimal~\cite{BST14}. This is a significant $\sqrt{d}$ improvement over the exponential mechanism. 

Exponential mechanism is a universally powerful tool in differential privacy.
However, nearly all of the previous works on DP-ERM rely on noisy gradient descent or its variants to achieve the significant $\sqrt{d}$ improvement over exponential mechanism under $(\eps,\delta)$-DP. 
One natural question is whether noisy gradient descent has some extra ability that exponential mechanism lacks or we didn't use exponential mechanism optimally in this setting. This brings us to the first question.
\begin{question}
Can we obtain the optimal empirical risk in \eqref{eqn:empirical_risk} under $(\eps,\delta)$-DP using exponential mechanism?
\end{question}

\paragraph{DP Stochastic Convex Optimization (DP-SCO)} Beyond the privacy guarantee and the empirical risk guarantee, another important guarantee is the generalization guarantee. Formally, we assume the users are sampled from an unknown distribution $\cP$ over convex functions. We define the loss function as
\begin{align}
\label{eq:DPSCO}
    \HF(x)=\E_{s\sim \cP}[f(x;s)].
\end{align}
We want to design a DP mechanism $\cM$ which outputs $x^{priv}$ given users $\cD = \{s_1,s_2,\dots,s_n\}$ independently sampled from $\cP$ and minimize the \emph{excess population loss}
\begin{equation}
\label{eqn:population_loss}
    \E_{\cM,\cD\sim \cP}[\HF(x^{priv})] - \HF(x^*)
\end{equation}
where $x^*$ is the minimizer of $\HF(x)$. We call the problem of minimizing the excess population loss in \eqref{eqn:population_loss} as DP Stochastic Convex Optimization (DP-SCO).
By a suitable modification of noisy stochastic gradient descent, \cite{bftt19,FKT20} show that one can achieve the optimal population loss of
\begin{equation}
\label{eqn:optimal_population_loss}
O\lp GD \lp\frac{1}{\sqrt{n}}+\frac{\sqrt{d\log(1/\delta)}}{\eps n}\rp\rp.
\end{equation}
\cite{bftt19} bounds the generalization error by showing that running SGD on smooth functions is stable and \cite{FKT20} proposes an iterative localization technique.
Note that only the algorithm for smooth functions in \cite{bftt19} can achieve both optimal empirical risk and optimal population loss at the same time, with the price of taking more gradient queries and loss of efficiency.
%$O\left(\min \left\{n^{3 / 2}, n^{5 / 2} / d\right\}\right)$ which uses  gradient queries and is less efficient than the linear rate achieved by \cite{FKT20}.
%algorithms for non-smooth functions does not simultaneously give the optimal empirically risk and the optimal population loss and 
It is unclear to us how one can obtain both using current techniques for non-smooth functions.
%\Yintat{make sure it is not a lie}. 
This brings us to the second question. 
\begin{question}
Can we achieve both the optimal empirical risk and the optimal population loss for non-smooth functions with the same algorithm?
\end{question}
%\Gopi{Cite the appropriate papers.}

\paragraph{Sampling} Without extra smoothness assumptions on $f$, currently, there is no optimally efficient algorithm for both problems. For example, with oracle access to gradients of $f$, the previous best algorithms for DP-SCO use:
\begin{itemize}
    \item $\widetilde{O}(n d)$ queries to $\nabla f(x;s)$ (by combining \cite{FKT20}, Moreau-Yosida regularization and cutting plane methods),
    \item $\widetilde{O}(\min(n^{3/2},n^{2}/\sqrt{d}))$ queries to $\nabla f(x;s)$ \cite{AFKT21},
    \item $\widetilde{O}(\min(n^{5/4}d^{1/8},n^{3/2}/d^{1/8}))$ queries to $\nabla f(x;s)$ \cite{KLL21}.
\end{itemize}
Combining these results, this gives an algorithm for DP-SCO that uses 
$$\widetilde{O}(\min(nd, n^{5/4}d^{1/8}, n^{3/2}/d^{1/8},n^{2}/\sqrt{d}))$$
many queries to $\nabla f(x;s)$. Although the information lower bound for non-smooth functions with the gradient queries is open, it is unlikely that the answer involves four different cases. 

In this paper, we focus on the function value query (zeroth order query) on $f(x;s)$. This query is weaker than gradient query as it obtains $d$ times less information. They are used in many practical applications such as clinical trials and ads placement when the gradient is not available and is also useful in bandit problems.
This brings us to the third question.

\begin{question}
Can we obtain an algorithm with optimal query complexity for DP-SCO for zeroth order query model?
\end{question}

\subsection{Our Contributions}
In this paper, we give a positive answer to all these questions using the \emph{Regularized Exponential Mechanism}. If we add an $\ell_2^2$ regularizer to $F$ and sample $x^{priv}$ from the density 
\begin{equation}
\label{eqn:our_mechanism}
\exp\lp-k\lp F(x;\cD)+\mu \norm{x}_2^2/2\rp\rp,
\end{equation}
then, for a suitable choice of $\mu$ and $k$, we recover the optimal excess risk in (\ref{eqn:optimal_empirical_risk}) for DP-ERM and optimal population loss in (\ref{eqn:optimal_population_loss}) for DP-SCO. Finally, we give an algorithm to sample $x^{priv}$ from the density \eqref{eqn:our_mechanism} with nearly optimal number of queries to $f(x;s)$ (See Figure \ref{fig:sample_runtime}). To the best of our knowledge, our algorithm is the first whose query complexity has \emph{polylogarithmic dependence} in both dimension and accuracy (in TV distance).

%Note that our mechanism is essentially how statisticians perform modeling: model the data by a log-density $f(x;s)$, add some regularizer, and learn the model parameters by a sample. Therefore, our result provides insights on the importance and advantages of a good regularizer, and a guideline on how to choose it. \Gopi{I didn't understand this paragraph.}

%Previously, sampling algorithms involves a large polynomial runtime for non-smooth objective (See Figure \ref{fig:sample_runtime}). Surprisingly, our algorithm is the first with the query complexity \emph{polylogarithmic dependence} in both dimension and accuracy and the first that is optimal in any function class we aware of.

Formally, our result is follows:
\begin{theorem}[DP-ERM, Informal]
Let $\cK$ be a convex set with diameter $D$ and $\{f(\cdot;s)\}$
be a family of convex functions on $\cK$ where $f(\cdot;s)-f(\cdot;s')$ is $G$-Lipschitz for all $s,s'$. Given a
database $\mathcal{D}=\{s_{1},s_{2},\cdots,s_{n}\}$, for any $\epsilon,\delta\in(0,\frac{1}{10})$,
\footnote{See Theorem~\ref{thm:DPERM} for general conclusions for all $\eps>0$}
the regularized exponential mechanism 
\[
x^{(priv)}\propto\exp\lp-k\cdot\lp\frac{1}{n}\sum_{i=1}^{n}f(x;s_{i})+\frac{\mu}{2}\|x\|_{2}^{2}\rp\rp
\]
is $(\epsilon,\delta)$-DP with expected excess empirical loss
\[
\frac{2GD\sqrt{d\log(1/\delta)}}{\epsilon n}
\]
for some appropriate choices of $k$ and $\mu$. Furthermore, if $f(\cdot;s)$ is $G$-Lipschitz for all $s$, 
we can sample $x^{(priv)}$ using $O(\frac{\eps^2n^2}{\log(1/\delta)}\log^2(\frac{nd}{\delta}))$ queries in expectation to the values of $f(x;s)$.
\end{theorem}

\begin{theorem}[DP-SCO, Informal]
Let $\cK$ be a convex set with diameter $D$ and $\{f(\cdot;s)\}$
be a family of convex functions on $\cK$ where $f(\cdot;s)-f(\cdot;s')$ is $G$-Lipschitz for all $s,s'$. 
Given a
database $\mathcal{D}=\{s_{1},s_{2},\cdots,s_{n}\}$ of samples from
some unknown distribution $\mathcal{P}$. For any $\epsilon,\delta\in(0,\frac{1}{10})$,\footnote{
See Theorem~\ref{thm:dpsco_impl} for general conclusions for all $\eps>0$.
}
the regularized exponential mechanism 
\[
x^{(priv)}\propto\exp\lp-k\cdot\lp\frac{1}{n}\sum_{i=1}^{n}f(x;s_{i})+\frac{\mu}{2}\|x\|_{2}^{2}\rp\rp
\]
is $(\epsilon,\delta)$-DP with expected excess population loss
\[
\frac{2GD}{\sqrt{n}}+\frac{2GD\sqrt{d\log(1/\delta)}}{\epsilon n}
\]
for some appropriate choice of $k$ and $\mu$. Furthermore, if $f(\cdot;s)$ is $G$-Lipschitz for all $s$, 
we can sample $x^{(priv)}$ using $O(\min\{\frac{\eps^2n^2}{\log(1/\delta)},nd\}\log^2(\frac{nd}{\delta}))$ queries in expectation to the
values of $f(x;s)$ and the expected number of queries is optimal up to logarithmic terms. 
\end{theorem}

For DP-SCO, we provide a nearly matching information-theoretic lower bound on the number of value queries (Section~\ref{sec:infolower}), proving the optimality of our sampling algorithm.
Moreover, when $f$ is already strongly convex, our proof shows the exponential mechanism (without adding a regularizer) itself simultaneously achieves both the optimal excess empirical risk and optimal population loss. 

In a {\em concurrent and independent} work, \cite{GTU22} study the DP properties of Langevin Diffusion, and provide optimal/best known private empirical risk and population loss under both pure-DP ($\delta=0$) and approximate-DP ($\delta>0$) constraints.
Utility/privacy trade-off of non-convex functions is also discussed.

%Unfortunately, if $f$ is just convex, we need to pick different $\mu$ and different $k$ for DP-ERM and DP-SCO. We are not sure how to obtain the optimal excess empirical risk and optimal population loss with the same algorithm and the same parameters. We leave this as an open problem.

% \begin{theorem}[Lower bound for DP-SCO]
% For any (non-private) algorithm which makes less than $O\lp\min\{\frac{\eps^2n^2}{\log(1/\delta)},nd\}\rp$ function value queries, there exists a convex domain $\cK\subset \R^d$ of diameter $D$, a distribution $\cP$ supported over $G$-Lipschitz linear functions $f(x;s)\defeq\langle x,s\rangle$ such that the output $\hx$ of the algorithm has excess population loss
% \begin{align*}
%     \E_{s\sim\cP}[\langle \hx,s\rangle]-\min_{x\in\cK}\E_{s\sim\cP}[\langle x,s\rangle]\geq \Omega\lp \frac{GD}{\sqrt{1+\log(n)/d}}\cdot \min \lc 1,\frac{\sqrt{\log(1/\delta)d}}{\eps n}+\frac{1}{\sqrt{n}}\rc \rp.
% \end{align*}

% \end{theorem}

 %Finally, we are not sure if the constant ($3$ and $4$ in the loss term) in the theorems above are optimal.  3 is not optimal, we can easily make the proof 2.0001

% \begin{algorithm2e}
% \SetKwInOut{Input}{input}
% \SetKwInOut{Output}{output}
% \caption{Regularized Exponential Mechanism \label{algo:regularized_exp}}
% \Input{Convex function $F$, parameters $k,\mu\in \R^{>0}$}
% \Output{ A sample $x^{priv}$ from the density $\pi_\cD(x) \propto \exp\lp-k\lp F(x;\cD)+\mu \norm{x}_2^2/2\rp\rp$}
% \end{algorithm2e}

\section{Techniques}

The main contribution of this paper is the discovery that adding regularization terms in exponential mechanism leads to optimal algorithms for DP-ERM and DP-SCO. For this, we develop some important tools that could be of independent interest. We now briefly discuss each of the main tools.

% Our proofs mainly rely on some less well-known techniques and involves surprising twists.\Gopi{This isn't a movie, lol. Need to use more formal language in papers.}
% In this section, we discuss some techniques that maybe of independent interest. 

\subsection{Gaussian Differential Privacy (GDP) of Regularized Exponential Mechanism}

To analyze the privacy of the regularized exponential mechanism, we need to bound the privacy curve between a strongly log-concave distribution and its Lipschitz perturbation in the exponent. \cite{MASN16} gave a nearly tight (up to constants) privacy guarantee of exponential mechanism if the distribution $\exp(-k F(x;\cD))$ satisfies Logarithmic Sobolev inequality (LSI). Since strongly log-concave distributions satisfy LSI, their result immediately gives the $(\epsilon,\delta)$-DP guarantee of our algorithm. However, this gives a sub-optimal privacy bound because it does not fully take advantage of the strongly log-concave property.

Instead, we show directly that the privacy curve between a strongly log-concave distribution and its Lipschitz perturbation in the exponent is upper bounded by the privacy curve of an appropriate Gaussian mechanism. This new proof uses the notion of tradeoff function introduced in~\cite{dong2019gaussian} and the isoperimetric inequality for strongly log-concave distribution.

\begin{theorem}
\label{thm:privacy_technical}
Given convex set $\cK\subseteq \R^d$ and $\mu$-strongly convex functions $F,\Tilde{F}$ over $\cK$. Let $P,Q$ be distributions over $\cK$ such that $P(x)\propto e^{-F(x)}$ and $Q(x)\propto e^{-\Tilde{F}(x)}$.
If $\Tilde{F}-F$ is $G$-Lipschitz over $\cK$, then for all $\eps>0$,
\begin{align*}
    \deltacurve{P}{Q}(\epsilon) 
    \leq \deltacurve{\cN\lp 0,1\rp}{\cN\lp\frac{G}{\sqrt{\mu}},1\rp}(\epsilon).
\end{align*}
\end{theorem}
This proves that the privacy curve for distinguishing between $P,Q$ is upper bounded the privacy curve of a Gaussian mechanism with sensitivity $G/\sqrt{\mu}$ and noise scale 1.

\paragraph{Tightness:} Note that Theorem~\ref{thm:privacy_technical} is completely tight because it contains the privacy of Gaussian mechanism as a special case. If $F(x)=\norm{x}_2^2/2$ and $\Tilde{F}(x)=\norm{x-a}_2^2/2$ for some $a\in \R^d$, then $\Tilde{F}(x)-F(x)=-\inpro{x}{a}+\norm{a}_2^2/2$ is $G$-Lipschitz with $G=\norm{a}_2$ and $F,\Tilde{F}$ are $1$-strongly convex. And $P=\cN(0,I_d)$ and $Q=\cN(a,I_d)$. Therefore:
$$\deltacurve{P}{Q}=\deltacurve{\cN(0,I_d)}{\cN(a,I_d)}=\deltacurve{\cN(0,1)}{\cN\lp\norm{a}_2,1\rp}$$ which is precisely the upper bound guaranteed by the theorem.

\subsection{Generalization Error of Sampling}
Many important and fundamental problems in machine learning, optimization and operations research are special cases of
SCO, and ERM is a classic and widely-used approach to solve it, though their relationships are not well-understood.
If one can solve the ERM problem optimally and get the exact optimal solution $x^*$ to minimizing $F(\cdot;\cD)$ (see Equation~\ref{eq:DPERM}), then \cite{SSSSS09} showed $x^*$ will also be a good solution to the SCO for strongly convex functions.
But in most situations, solving ERM optimally costs too much or even impossible. 
Can we find a approximately good solution to ERM and hope that it is also a good solution for SCO?
\cite{Fel16} provides a negative answer and shows there is no good uniform convergence between $F(\cdot;\cD)$ and $\HF$, that is there always exists $x\in\cK$ such that $|F(x;\cD)-\HF(x)|$ is large.
This fact forces us to find approximate solution to ERM with very high accuracy, which makes the algorithms inefficient.

Prior works proposed a few interesting ways to overcome this difficulty, such as the uniform stability in \cite{HRS16} and the iterative localization technique in \cite{AFKT21}.
Roughly speaking, uniform stability means that if running algorithms on neighboring datasets lead to similar output distributions, then the generalization error of the ERM algorithm is bounded.
Thus a good solution to ERM obtained by a stable algorithm is also a good solution for SCO.
\cite{bftt19} makes use of the stability of running SGD on smooth functions to get a tight bound on the population loss for DP-SCO.

Recall $F(x;\cD)$ and $\HF(x)$ are defined in Equation~\eqref{eq:DPERM} and \eqref{eq:DPSCO} respectively. 
Our result enriches the toolbox of bounding the generalization error and provides new insights for this problem.
\begin{theorem}
Suppose $\{f_i\}$ is a family of $\mu$-strongly convex functions over $\cK$ and $f_i-f_{i'}$ is $G$-Lipschitz for any two functions $f_i,f_{i'}$ in the family.
For any $k>0$ and suppose the $n$ samples in data set $\cD$ are drawn i.i.d from the underlying distribution, then by sampling $x^{(sol)}$ from density $\propto e^{-kF(x^{(sol)};\cD)}$, the population loss satisfies
\begin{align*}
    \E[\HF(x^{(sol)})]-\min_{x\in\cK}\HF(x)\leq \frac{G^2}{\mu n}+ \frac{d}{k}.
\end{align*}
\end{theorem}

Considering two neighboring datasets $\cD$ and $\cD'$, our result is based on bounding the Wasserstein distance between the distributions proportional to $e^{-kF(x;\cD)}$ and $e^{-kF(x;\cD')}$, which means the sampling scheme is stable and leads to the $\frac{G^2}{\mu n}$ term in generalization error.
% The proof makes use of the fact that any divergence measure that decreases under post-processing such as KL divergence.
The other term $\frac{d}{k}$ is excess empirical loss of the sampling mechanism.
One advantage of our result is that it works for both smooth and non-smooth functions. 
Moreover, we may choose the value $k$ carefully and get a solution with both optimal empirical loss and optimal population loss.

% Write the statement of generalization error here (Theorem 6.10, make it as a restatable, change the statement such that it is just about generalization, but not about DP so that other people can use).
% Also, explain the proof Theorem 6.10.

\subsection{Non-smooth Sampling and DP Convex Optimization }
Implementing the exponential mechanism involves sampling from a log-concave distribution.
When the negative log-density function $F$ is smooth, i.e. the gradient of $F$ is Lipschitz,
there are many efficient algorithms for this sampling tasks
such as \cite{D17,LSV18,MMW+19,CV19,DMM19,shen2019randomized,CDW+20,LST20}.
For example, if $F=\frac{1}{n}\sum_{i=1}^{n}f_{i}$
and each $f_{i}$ is $1$-strongly convex with $\kappa$-Lipschitz
gradient,\footnote{For convenience, we used $f_{i}$ to denote the function $f(\cdot;s_i)$ in this and Section \ref{sec:sampling}.} we can sample $x\sim\exp(-F(x))$ in $\widetilde{O}(n+\kappa\max(d,\sqrt{nd})\log(1/\delta))$
iterations with $\delta$ error in total variation distance and each iteration involves computing one $\nabla f_{i}(x)$ \cite{LST21}.
Note that this is nearly linear time when $n\gg\kappa^2d$ and the $\delta$ error in
total variation distance can be translated to an extra $\delta$ error in the $(\epsilon,\delta)$-DP
guarantee.

\begin{center}
\begin{figure}[ht]
\begin{centering}
\begin{tabular}{|c|c|c|c|}
\hline 
 & Complexity & Oracle & Guarantee\tabularnewline
\hline 
\hline 
\cite{BST14} & $d^{O(1)}$ & $F(x)$ & $\mathrm{D}_{\infty}\leq\epsilon$\tabularnewline
\hline 
\cite{CDJB20} & $G^{O(1)}d^{5/2}/\epsilon^{4}$ & $\nabla F(x)$ & $\mathrm{W}_{2}\leq\delta$\tabularnewline
\hline 
\cite{JLLV21} + \cite{C21} & $d^{3}$ & $F(x)$ & $\mathrm{TV}\leq\delta$\tabularnewline
\hline 
\cite{GT20}& $\frac{\alpha^{2}G^{4}d}{\epsilon^{2}}$ & $\nabla F(x)$ & $\mathrm{D}_{\alpha}\leq\epsilon$\tabularnewline
\hline 
\cite{LC21} & $\frac{G^{2}}{\delta}$ & $\nabla F(x)$ & $\mathrm{TV}\leq\delta$\tabularnewline
\hline 
This & $G^{2}$ & $f_{i}(x)$ & $\mathrm{TV}\leq\delta$\tabularnewline
\hline 
\end{tabular}
\par\end{centering}
\caption{The complexity of sampling from $\exp(-F(x))$ where $F=\frac{1}{n}\sum_{i}f_{i}$
is $1$-strongly convex and $f_{i}$ are $G$-Lipschitz and convex.
For applications in differential privacy, $\epsilon$ is a constant
and $\delta=n^{-\Theta(1)}$. Polylogarithmic terms are omitted. Only the last result uses the summation structure and queries only one $f_{i}$ each step. \label{fig:sample_runtime}}
\end{figure}
\vspace{-7mm}
\par\end{center}

Unfortunately, when the functions $f_i$ are only Lipschitz but not smooth, this problem is more difficult.
In Table \ref{fig:sample_runtime}, we summarize some existing results on this topic. They use different guarantees such as Renyi divergence $\mathrm{D}_{\alpha}$ of order $\alpha$, Wasserstein distance $\mathrm{W}_{2}$ and total variation distance $\mathrm{TV}$ (defined in subsection~\ref{sec:dis_measure}). For applications in differential privacy, we need either polynomially small $\mathrm{W}_{2}$ or $\mathrm{TV}$ distance, or $\epsilon$ small $\mathrm{D}_{\alpha}$ distance.

All previous results for non-smooth function use oracle access to $F$ or $\nabla F$ (instead of $f_i$) and have iterative complexity at least $d$ iterations for $\mathrm{W}_{2}$ or TV distance smaller than $1/d$. Because of this, our algorithm is significantly faster than
the previous algorithms and can handle the case when $F$ is expectation of (infinitely many)  $f_i$ directly.
For example, to get the optimal private empirical loss with typical settings where $\epsilon=\Theta(1)$
and $\delta=1/n^{\Theta(1)}$, the previous best samplers use $\widetilde{O}(n^{4}d)$ many queries to $\nabla f_{i}(x)$ by \cite{GT20} or $\widetilde{O}(nd^{3})$
many queries to $f_{i}(x)$ by combining \cite{JLLV21} and \cite{C21}. 
% The extra $n$ factors are because of the need to choose a large scaling factor $k$ in the distribution $\exp(-kF)$ to make it private. 
In comparison, our algorithm only takes $\widetilde{O}(n^{2})$
many $f_{i}(x)$. 

Our result is based on the alternating sampler proposed in \cite{LST21} and a new rejection sampling scheme.

%{[}Copy the main statement here{]}
%To sample from $\exp(-F(x))$ where $F\defeq\E_{i\in I}f_i+\psi(x)$ where $\psi(x)$ is $\mu$-strongly convex and each $f_i$ is $G$-Lipschitz and convex, our algorithm needs $\Tilde{O}(\frac{G^2}{\mu}\log^2(1/\delta))$ (ignore other logarithmic terms) iterations, and in each iteration it queries $O(1)$ many values of $f_i$ in expectation and finally outputs a sample, total variation distance between whose distribution and the $\exp(-F)$ is bounded by $\delta$.
%{[}restatable of the sampling result{]}
\begin{theorem}
%\label{thm:sampler}
Given a $\mu$-strongly convex function $\psi(x)$ defined on a convex set $\cK \subseteq \R^{d}$ and $+\infty$ outside. Given a family of $G$-Lipschitz convex functions $\{f_{i}(x)\}_{i\in I}$ defined on $\cK$ and an initial point $x_0\in \cK$.
Define the function $\widehat{F}(x)=\E_{i\in I}f_{i}(x)+\psi(x)$ and 
the distance $D=\|x_{0}-x^{*}\|_{2}$ for some $x^{*}=\arg\min_{x\in\cK}\HF(x)$.
For any $\delta\in(0,1/2)$, we can generate a random point $x$ that
has $\delta$ total variation distance to the distribution proportional to $\exp(-\widehat{F}(x))$ in
\[
T:=\Theta\lp\frac{G^{2}}{\mu}\log^{2}\lp\frac{G^{2}(d/\mu+D^{2})}{\delta}\rp\rp\text{ steps}.
\]
Furthermore, each steps accesses only $O(1)$ many $f_{i}(x)$ and samples from $\exp(-\psi(x) - \frac{1}{2\eta} \|x-y\|^2_2)$ for $O(1)$ many $y$
in expectation with $\eta = \Theta(G^{-2}/\log(T/\delta))$.
\end{theorem}

\section{Preliminaries}

\subsection{Differential Privacy}

A DP algorithm $\M$ usually satisfies a collection of $(\eps,\delta)$-DP guarantees for each $\epsilon$, i.e., for each $\epsilon$ there exists some smallest $\delta$ for which $\M$ is $(\eps,\delta)$-DP. By collecting all of them together, we can form the privacy curve or privacy profile which fully characterizes the privacy of a DP algorithm.
\begin{definition}[Privacy Curve]
Given two random variables $X,Y$ supported on some set $\Omega$, define the privacy curve $\delta(X\|Y):\R_{\geq 0}\rightarrow [0,1]$ as:
\begin{align*}
    \delta(X\|Y)(\epsilon)=\sup_{{S}\subset \Omega} \Pr[Y\in {S}]-e^{\epsilon}\Pr[X\in {S}].
\end{align*}
\end{definition}

One can explicitly calculate the privacy curve of a Gaussian mechanism as
\begin{equation}
\label{eqn:Gaussian_privacycurve}
   \deltacurve{\cN(0,1)}{\cN(s,1)}(\eps)= \Phi\lp -\frac{\eps}{s}+\frac{s}{2} \rp - e^\eps \Phi\lp -\frac{\eps}{s}-\frac{s}{2} \rp 
\end{equation}
 where $\Phi(\cdot)$ is the Gaussian cumulative distribution function (CDF)~\cite{BalleW18}. 
%  One can also prove the following upper bound on the privacy curve by using the upper bound $\Phi(-t)\leq \exp(-t^2/2)/2$ for $t\ge 0,$ % Double check in a computer, the tight bound has a extra 1/2. (I find it in wiki)
% \begin{equation}
% \label{eqn:Gaussian_privacycurve_approx}
%   \deltacurve{\cN(0,1)}{\cN(s,1)}(\eps)  \le \frac{1}{2} \exp\lp-\frac{1}{2}\lp\frac{\eps}{s}-\frac{s}{2}\rp^2\rp  \text{ if } \eps \geq \frac{s^2}{2}.
% \end{equation}

We say a differentially private mechanism $\M$ has privacy curve $\delta:\R_{\ge0}\rightarrow[0,1]$ if for every $\epsilon\geq0$, $\M$ is $(\epsilon,\delta(\epsilon))$-differentially private, i.e., $\delta(\M(\cD)\| \M(\cD'))(\epsilon)\leq \delta(\eps)$ for all neighbouring databases $\cD,\cD'$.
We will also need the notion of tradeoff function introduced in~\cite{dong2019gaussian} which is an equivalent way to describe the privacy curve $\delta(P\|Q)$.
\begin{definition}[Tradeoff function]
Given two (continuous) distributions $P,Q$, we define the trade-off function\footnote{Tradeoff curves in~\cite{dong2019gaussian} are defined using type I and type II errors. The definition given here is equivalent to their definition for continuous distributions.} $T(P\|Q):[0,1]\to [0,1]$ as $$T(P\|Q)(z)= \inf_{S:P(S)=1-z}Q(S).$$

It is easy to compute explicitly the tradeoff function for Gaussian mechanism~\cite{dong2019gaussian},
\begin{equation}
    \label{eqn:gaussian_tradeoff}
    T(\cN(0,1)\|\cN(s,1))(z)=\Phi(\Phi^{-1}(1-z)-s).
\end{equation}
Note that perfect privacy is equivalent to the tradeoff function $\mathrm{Id}(z)=1-z$ and the closer a tradeoff function is to $\mathrm{Id}$, better the privacy. The tradeoff function $T(P\|Q)$ and the privacy curve $\delta(P\|Q)$ are related via convex duality. Therefore to compare privacy curves, it is enough to compare tradeoff curves.

\begin{proposition}[\cite{dong2019gaussian}]
\label{prop:delta_tradeoff}
$\delta(P\|Q)\le \delta(P'\|Q')$ iff $T(P\|Q)\ge T(P'\|Q')$
\end{proposition}

\end{definition}

\subsection{Optimization}
Here we collect some properties of functions which are useful for optimization and sampling.
\begin{definition}[$L$-Lipschitz Continuity]
A function $f:{\cal K}\rightarrow \R$ is $L$-Lipschitz continuous over the domain ${\cal K}\subset \R^{d}$ if the following holds for all $\omega,\omega'\in {\cal K}:|f(\omega)-f(\omega')|\leq L\|\omega-\omega'\|_2$. 
\end{definition}

% never used
%\begin{definition}[$\beta$-Smoothness]
%A function $f:{\cal K}\rightarrow \R$ is $\beta$-smooth over the domain ${\cal K}\subset \R^{d}$ if for all $\omega,\omega'\in{\cal K}$, $\|\nabla f(\omega)-\nabla f(\omega')\|_2\leq \beta \|\omega-\omega'\|_2$.
%\end{definition}

\begin{definition}[$\mu$-Strongly convex]
A differentiable function $f:\cK\rightarrow \R$ is called strongly convex with parameter $\mu>0$ if ${\cal K}\subset \R^{d}$ is convex and the following inequality holds for all points $\omega,\omega'\in {\cal K}$,
% one def is enough here
%\[
%\langle \nabla f(\omega)-\nabla f(\omega'), \omega-\omega' \rangle\geq \mu \|\omega-\omega'\|_2^2.
%\]
%Equivalently,
\[
f(\omega')\geq f(\omega) +\inpro{\nabla f(\omega)}{\omega'-\omega}+\frac{\mu}{2}\|\omega'-\omega\|_2^2.
\]
\end{definition}

\begin{definition}[Log-concave measure and density]
A density function $f:\cK \rightarrow \R_{\geq 0}$ is log-concave if $\int_{\cK} f(x) dx = 1$ and $f(x) = \exp(-F(x))$ for some convex function $F$. We call $f$ is $\mu$-strongly log-concave if $F$ is $\mu$-strongly convex. Similarly, we call $\pi$ a log-concave measure if its density function is log-concave, and we call $\pi$ is a $\mu$-strongly log-concave measure if its density function is $\mu$-strongly log-concave.
\end{definition}

\subsection{Distribution Distance and Divergence}
We present some distribution distances or divergences mentioned or used in this work.
\label{sec:dis_measure}
\begin{definition}{\cite[R{\'e}nyi Divergence]{Ren61}}
Suppose $1<\alpha<\infty$ and $\pi,\nu$ are measures with $\pi\ll\nu$.
The R{\'e}nyi divergence of order $\alpha$ between $\pi$ and $\nu$ is defined as
\begin{align*}
    \mathrm{D}_{\alpha}(\pi\|\nu)=\frac{1}{\alpha}\log\int\lp\frac{\pi(x)}{\nu(x)}\rp^{\alpha}\nu(x)\d x.
\end{align*}
We follow the convention that $\frac{0}{0}=0$. R{\'e}nyi Divergence of orders $\alpha=1,\infty$ are defined by continuity.
For $\alpha=1$, the limit in R{\'e}nyi Divergence equals to the Kullback-Leibler divergence of $\pi$ from $\nu$, which is defined as following:
\end{definition}

\begin{definition}[Kullback–Leibler divergence]
The Kullback–Leibler divergence between probability measures $\pi$ and $\nu$ is defined by 
\begin{align*}
    \mathrm{D}_{KL}(\pi\|\nu)=\int  \log\lp\frac{\pi}{\nu}\rp\d \pi.
\end{align*}
\end{definition}

\begin{definition}[Wasserstein distance]
Let $\pi,\nu$ be two probability distributions on $\R^d$.
The second Wasserstein distance $\mathrm{W}_{2}$ between $\pi$ and $\nu$ is defined by
\begin{align*}
    \mathrm{W}_2(\pi,\nu)=\big( \inf_{\gamma\in\Gamma(\pi,\nu)}\int_{\R^d\times\R^d}\|x-y\|_2^2\d \gamma(x,y) \big)^{1/2},
\end{align*}
where $\Gamma(\pi,\nu)$ is the set of all couplings of $\pi$ and $\nu$.
\end{definition}

\begin{definition}[Total variation distance]
The total variation distance between two probability measures $\pi$ and $\nu$ on a sigma-algebra $\mathcal{F}$ of subsets of the sample space $\Omega$ is defined via
\begin{align*}
    \mathrm{TV}(\pi,\nu)=\sup_{S\in \mathcal{F}}|\pi(S)-\nu(S)|.
\end{align*}
\end{definition}

\subsection{Isoperimetric Inequality for Strongly Log-concave Distributions}
The cumulative distribution function (CDF) of one-dimensional standard Gaussian distribution will be denoted by $\Phi(x)=\Pr_{y\sim\cN(0,1)}[y\le x]$. The following Lemma relates the expanding property of log-concave measures with $\Phi$.
%Note that by the symmetry of the Gaussian distribution, $\Phi(x)+\Phi(-x)=1$ and $\Phi^{-1}(1-z)=-\Phi^{-1}(z).$
\begin{proposition}[Theorem 1.1. in \cite{Led99}]
\label{prop:isoperimetry}
Let $\pi$ be a $\mu$-strongly log-concave measure supported on a convex set $\cK \subseteq \R^d$. Let $A\subset\cK$ by any subset such that $\pi(A)=z$. For any point $x\in\R^d$, define $d(x,A)=\inf_{y\in A}\|x-y\|_2$. Let $A_r = \lc x: d(x,A)\le r\rc$. Then if $A_r\subseteq\cK$, for every $r\ge 0$, $$\pi(A_r) \ge \Phi(\Phi^{-1}(z)+r\sqrt{\mu}).$$
\end{proposition}
The property above implies the concentration of Lipschitz functions over log-concave measures.
\begin{corollary}
\label{cor:isoperimetry_lip}
Let $\pi$ be a $\mu$-strongly log-concave measure supported on a convex set $\cK\subseteq\R^d$. Suppose $\alpha:\cK\to \R$ is $G$-Lipschitz. For $z\in [0,1]$, define $m(z)\in \R$ such that $\Pr_{x\sim \pi}[\alpha(x)\le m(z)]=z$. Then for every $r\ge 0$,
$$\Pr_{x\sim \pi}[\alpha(x) \ge m(z)+r] \le \Phi\lp \Phi^{-1}(1-z)-\frac{r\sqrt{\mu}}{G}\rp,$$
$$\Pr_{x\sim \pi}[\alpha(x) \le m(z)-r] \le \Phi\lp \Phi^{-1}(z)-\frac{r\sqrt{\mu}}{G}\rp.$$
\end{corollary}
\begin{proof}
Fix some $z\in [0,1]$. Let $A=\{x\in \cK:\alpha(x)\le m(z)\}$, so $\pi(A)=z$. Let $A_r=\{x:d(x,A)\le r\}.$ Since $\alpha$ is $G$-Lipschitz, $\alpha(x)\ge m(z)+r$ implies that $d(x,A)\ge r/G.$ Therefore $\lc x: \alpha(x) \ge m(z)+r\rc \subset \lc x: d(x,A) \ge r/G\rc = \overline{A_{r/G}}$ and so 
\begin{align*}
\Pr_{x\sim \pi}[\alpha(x) \ge m(z)+r] &\le \pi(\overline{A_{r/G}}) \\
&= 1- \pi(A_{r/G})\\
&\le 1 - \Phi\lp\Phi^{-1}(z)+\frac{r\sqrt{\mu}}{G}\rp\\
&= \Phi\lp-\Phi^{-1}(z)-\frac{r\sqrt{\mu}}{G}\rp.
\end{align*}
We obtain the other inequality by applying the above inequality to $-\alpha(x).$
\end{proof}

% \begin{theorem}[Sampling]

% \end{theorem}

%\input{exponential}
%\input{continuous}
% \input{log-sobolev}

\section{GDP of Regularized Exponential Mechanism}
In this section, we prove our DP result (Theorem~\ref{thm:privacy_technical}). 
The proof uses the isoperimetric inequality for strongly log-concave measures~\cite{Led99}. Intuitively, the privacy loss random variable will be $G$-Lipschitz under the hypothesis and isoperimetric inequality implies that any Lipschitz function will be as concentrated as a Gaussian with appropriate standard deviation. This allows us compare the privacy curve $\deltacurve{P}{Q}$ to that of a Gaussian mechanism. In our proof, it is actually more convenient to compare tradeoff curves ($\tradeoff{P}{Q}$) which are equivalent to privacy curves via convex duality (Proposition~\ref{prop:delta_tradeoff} and Theorem~\ref{thm:privacy_technical}).

%By Proposition~\ref{prop:delta_tradeoff}, Theorem~\ref{thm:privacy_technical} is equivalent to the following theorem.
%The proof is based on the Isoperemetric inequality, which presents the concentration property of strongly log-concave measure and the worst case is exactly the Gaussian measure.
%Thus the privacy curve can be bounded by the one of Gaussian mechanism.

\begin{theorem}
\label{thm:privacy_technical_tradeoff}
Given convex set $\cK\subseteq \R^d$ and $\mu$-strongly convex functions $F,\Tilde{F}$ over $\cK$. Let $P,Q$ be distributions over $\cK$ such that $P(x)\propto e^{-F(x)}$ and $Q(x)\propto e^{-\Tilde{F}(x)}$.
If $\Tilde{F}-F$ is $G$-Lipschitz over $\cK$, then for all $z\in[0,1]$,
\begin{align*}
    \tradeoff{P}{Q}(z) 
    \ge \tradeoff{\cN\lp 0,1\rp}{\cN\lp\frac{G}{\sqrt{\mu}},1\rp}(z).
\end{align*}
\end{theorem}

% \begin{proof}
% Let $\alpha(x)=\Tilde{F}(x)-F(x)+C$ for some constant $C\in \R$ so that $Q(x)=e^{-\alpha(x)} P(x)$. We have $T(P\|Q)(z)=\inf_{S:P(S)=z} Q(S)$. Note that the infimum is achieved when we choose 
% $$S=\lc x:\log\lp \frac{Q(x)}{P(x)}\rp = -\alpha(x) \le -m(z)\rc$$ for some $m(z)$ chosen such that 
% $$P(S)=\Pr_{x\sim P}[\alpha(x)\ge m(z)]=z.$$ Therefore:
% \begin{align*}
% T(P\|Q)(z)&= \int_{x: \alpha(x)\ge m(z)} Q(x) dx\\
% &= \int_{x: \alpha(x)\ge m(z)} e^{-\alpha(x)} P(x) dx\\
% &= \int_{t=m(z)}^\infty e^{-t} \lp -\frac{d\Pr_{x\sim P}\lb\alpha(x)\ge t\rb}{dt}\rp dx\\
% &=  \left.-e^{-t}\Pr_{x\sim P}\lb\alpha(x)\ge t\rb\right\vert_{m(z)}^\infty - \int_{t=m(z)}^\infty e^{-t} \Pr_{x\sim P}\lb\alpha(x)\ge t\rb dx\\
% &=  ze^{-m(z)} - e^{-m(z)}\int_{t=0}^\infty e^{-t} \Pr_{x\sim P}\lb\alpha(x)\ge t+m(z)\rb dx\\
% \end{align*}
% By isoperimetric inequality\Gopi{Refer to the theorem}, we have
% \begin{align*}
% \Pr_{x\sim P}\lb\alpha(x)\ge t+m(z)\rb &\le \Pr[N(0,1)\ge t+\Phi^{-1}(1-z)]\\
% &=1-\Phi\lp t+ \Phi^{-1}(1-z)\rp\\
% &=\Phi(\Phi^{-1}(z)-t) \tag{$\Phi^{-1}(1-z)=-\Phi^{-1}(z)$ and $\Phi(-x)+\Phi(x)=1$}.
% \end{align*}
% Plugging this into the expression for $T(P\|Q)$, we get:
% \begin{align*}
%     T(P\|Q)&\ge ze^{-m(z)} - e^{-m(z)}\int_{t=0}^\infty e^{-t} \Phi(\Phi^{-1}(z)-t) dx\\
%     &= ze^{-m(z)} - e^{-m(z)}\lp z - \exp\lp\frac{1}{2}-\Phi^{-1}(z)\rp \Phi(\Phi^{-1}(z)-1) \rp \tag{\Gopi{Integrate by parts}}\\
%     &= \exp\lp\frac{1}{2}-\Phi^{-1}(z)-m(z)\rp \Phi(\Phi^{-1}(z)-1)\\
% \end{align*}

% \end{proof}

\begin{proof}
% WLOG, we can assume $G=\mu=1$.
% One can easily bring in specific parameters with the same arguments.
% \Gopi{Explain why}
Let $\gamma=G/\sqrt{\mu}.$
Let $\alpha(x)=\Tilde{F}(x)-F(x)$ so that $Q(x)\propto e^{-\alpha(x)} P(x)$. Recall that we have $T(P\|Q)(z)=\inf_{S:P(S)=1-z} Q(S)$. Note that the infimum is achieved when we choose 
$S=\lc x\in\cK: \alpha(x) \ge m(z)\rc$ for some $m(z)$ chosen such that 
$P(S)=\Pr_{x\sim P}[\alpha(x)\ge m(z)]=1-z$ (Neyman-Pearson lemma). Therefore:
\begin{align*}
T(P\|Q)(z)&= \int_{x\in S} Q(x) \d x\\
&= \frac{\int_{x\in S} e^{-\alpha(x)} P(x) \d x}{\int_{x\in\cK} e^{-\alpha(x)} P(x) \d x}\\
%&= \frac{\E_P[e^{-\alpha}\Ind_S]}{\E_P[e^{-\alpha}\Ind_S] + \E_P[e^{-\alpha}\Ind_{\barS}]}\\
&= \lp 1+\frac{\E_P[e^{-\alpha}\Ind_{\barS}]}{\E_P[e^{-\alpha}\Ind_S]}\rp^{-1}
\end{align*}
% \begin{align*}
%     T(P\|Q)(z)&=\lp 1+\frac{\E_P[e^{-\alpha}\Ind_{\barS}]}{\E_P[e^{-\alpha}\Ind_S]}\rp^{-1}.
% \end{align*}
We will now lower bound $\E_P[e^{-\alpha}\Ind_S]$. Let the random variable $Y=\alpha(x)$ where $x\sim P.$ Let $f_Y(\cdot)$ be the PDF of $Y$.
\begin{align*}
\E_P[e^{-\alpha(x)}\Ind_S]&=\int_{x: \alpha(x)\ge m(z)}e^{-\alpha(x)}P(x) \d x=\E[e^{-Y}\Ind(Y\ge m(z))]=\int_{m(z)}^\infty e^{-t} f_Y(t) dt\\
&= \int_{t=0}^\infty e^{-t-m(z)} \lp -\frac{\d\Pr_{x\sim P}\lb\alpha(x)\ge t+m(z)\rb}{\d t}\rp \d t\\
&=  e^{-m(z)}\lp \left.-e^{-t}\Pr_{x\sim P}\lb\alpha(x)\ge t+m(z)\rb\right\vert_{0}^\infty - \int_{t=0}^\infty e^{-t} \Pr_{x\sim P}\lb\alpha(x)\ge t+m(z)\rb \d t\rp\\
&=  (1-z)e^{-m(z)} - e^{-m(z)}\int_{t=0}^\infty e^{-t} \Pr_{x\sim P}\lb\alpha(x)\ge t+m(z)\rb \d t\\
&\ge (1-z)e^{-m(z)} - e^{-m(z)}\int_{t=0}^\infty e^{-t} \Phi(\Phi^{-1}(1-z)-t/\gamma) \d t \tag{Corollary~\ref{cor:isoperimetry_lip}}\\
&= (1-z)e^{-m(z)} - e^{-m(z)}\lp (1-z) - \exp\lp\frac{\gamma^2}{2}-\Phi^{-1}(1-z)\gamma\rp \Phi(\Phi^{-1}(1-z)-\gamma) \rp \tag{Claim \ref{claim:some_gaussian_integrals}}\\
&= \exp\lp\frac{\gamma^2}{2}+\Phi^{-1}(z)\gamma-m(z)\rp \Phi(-\Phi^{-1}(z)-\gamma)
\end{align*}

%%%%%%%%%%%%%%%%%%%%%%%%%%%%%%%%%%%%%%%%%%%%%

We will now upper bound $\E_P[e^{-\alpha}\Ind_{\barS}]$ in a similar way.
\begin{align*}
\E_P[e^{-\alpha(x)}\Ind_{\barS}]&=\int_{x: \alpha(x)\le m(z)}e^{-\alpha(x)}P(x) \d x\\
&= \int_{t=0}^{\infty} e^{-m(z)+t} \lp -\frac{d\Pr_{x\sim P}\lb\alpha(x)\le m(z)-t\rb}{dt}\rp \d t\\
&=  e^{-m(z)}\lp\left.-e^{t}\Pr_{x\sim P}\lb\alpha(x)\le m(z)-t\rb\right\vert_{0}^\infty + \int_{t=0}^\infty e^{t} \Pr_{x\sim P}\lb\alpha(x)\le m(z)-t\rb \d t \rp\\
&=  ze^{-m(z)} + e^{-m(z)}\int_{t=0}^\infty e^{t} \Pr_{x\sim P}\lb\alpha(x)\le m(z)-t\rb \d t\\
&\le ze^{-m(z)} + e^{-m(z)}\int_{t=0}^\infty e^{t} \Phi(\Phi^{-1}(z)-t/\gamma) \d t \tag{Corollary~\ref{cor:isoperimetry_lip}}\\
    &= ze^{-m(z)} + e^{-m(z)}\lp -z + \exp\lp\frac{\gamma^2}{2}+\Phi^{-1}(z)\gamma\rp \Phi(\Phi^{-1}(z)+\gamma) \rp \tag{Claim~\ref{claim:some_gaussian_integrals}}\\
    &= \exp\lp\frac{\gamma^2}{2}+\Phi^{-1}(z)\gamma-m(z)\rp \Phi(\Phi^{-1}(z)+\gamma)
\end{align*}

Combining the two bounds, we get:
\begin{align*}
T(P\|Q)(z)&= \lp 1+\frac{\E_P[e^{-\alpha}\Ind_{\barS}]}{\E_P[e^{-\alpha}\Ind_S]}\rp^{-1}\\
&\ge \lp 1 + \frac{\Phi(\Phi^{-1}(z)+\gamma)}{\Phi(-\Phi^{-1}(z)-\gamma)} \rp^{-1}\\
&= \Phi(-\Phi^{-1}(z)-\gamma) \tag{Using $\Phi(x)+\Phi(-x)=1$}\\
&= \tradeoff{N(0,1)}{N(\gamma,1)} \tag{Eqn (\ref{eqn:gaussian_tradeoff})}.
\end{align*}
\end{proof}

We finish by calculating the integrals that showed up in the proof.
\begin{claim}
\label{claim:some_gaussian_integrals}
$$\int_0^\infty e^{-t} \Phi\lp a-\frac{t}{\gamma}\rp \d t = \Phi(a)-e^{\frac{\gamma^2}{2}-a\gamma}\Phi(a-\gamma)$$
$$\int_0^\infty e^{t} \Phi\lp a-\frac{t}{\gamma}\rp  \d t = -\Phi(a)+e^{\frac{\gamma^2}{2}+a\gamma}\Phi(a+\gamma)$$
\end{claim}
\begin{proof}
\begin{align*}
    \int_0^\infty e^{-t} \Phi(a-t/\gamma) \d t &= \left.-e^{-t}\Phi(a-t/\gamma)\right\vert_0^\infty - \int_0^\infty e^{-t} \frac{e^{-(a-t/\gamma)^2/2}}{\gamma\sqrt{2\pi}} \d t\\
     &= \Phi(a) - \int_0^\infty e^{\gamma^2/2-a\gamma} \frac{e^{-(t-(\gamma a-\gamma^2))^2/2}}{\gamma\sqrt{2\pi}} \d t\\
     %&= \Phi(a) -  e^{\gamma^2/2-a\gamma} \lp 1- \Phi(-(a-\gamma))\rp\\
     &= \Phi(a) -  e^{\gamma^2/2-a\gamma} \Phi(a-\gamma).
\end{align*}
\begin{align*}
    \int_0^\infty e^{t} \Phi(a-t/\gamma) \d t &= \left.e^{t}\Phi(a-t/\gamma)\right\vert_0^\infty + \int_0^\infty e^{t} \frac{e^{-(a-t/\gamma)^2/2}}{\gamma\sqrt{2\pi}} \d t\\
     &= -\Phi(a) + \int_0^\infty e^{\gamma^2/2+a\gamma} \frac{e^{-(t-(a\gamma+\gamma^2))^2/2\gamma^2}}{\gamma\sqrt{2\pi}} \d t\\
     %&= -\Phi(a) +  e^{\gamma^2/2+a\gamma} \lp 1- \Phi(-a-\gamma)\rp\\
     &= -\Phi(a)  +  e^{\gamma^2/2+a\gamma} \Phi(a+\gamma).
\end{align*}
\end{proof}

As a corollary to Theorem~\ref{thm:privacy_technical_tradeoff}, we can bound any divergence measure that decreases under post-processing such as Renyi divergence or KL divergence. In particular, this also implies Renyi Differential Privacy~\cite{Mir17} of our algorithm.

\begin{corollary}
\label{cor:divergence_privacy}
Suppose $F,\TF$ are two $\mu$-strongly convex functions over $\cK\subseteq \R^d$, and $F-\TF$ is $G$-Lipschitz over $\cK$.
For any $k>0$, if we let $P\propto e^{-kF}$ and $Q\propto e^{-k\TF}$ be two probability distributions on $\cK$, then we have
\begin{align*}
    \mathrm{D}(P\|Q)\leq \mathrm{D}\lp\cN(0,1)\|\cN\lp \frac{G\sqrt{k}}{\sqrt{\mu}},1\rp\rp
\end{align*}
for any divergence measure $\mathrm{D}$ which decreases under post-processing. In particular, $$\mathrm{D}_{\alpha}(P\|Q)\le \frac{\alpha k G^2}{2\mu} \text{ and }\mathrm{D}_{KL}(P\|Q)\le \frac{kG^2}{2\mu}.$$
\end{corollary}

\begin{proof}
% Recall in Theorem~\ref{thm:privacy_technical_tradeoff}, we proved for all $z\in[0,1]$, one has
% \begin{align*}
%     \tradeoff{\pi}{\nu}(z)\ge \tradeoff{\cN\lp 0,1\rp}{\cN\lp\frac{G\sqrt{k}}{\sqrt{\mu}},1\rp}(z).
% \end{align*}

By Theorem 2.10 in \cite{dong2019gaussian}, if $T(P\|Q) \ge T(X\|Y)$, then there exists a randomized algorithm $M$ such that $M(X)=P$ and $M(Y)=Q$. Therefore for any divergence measure which decreases under post-processing we have,
$$\mathrm{D}(P\|Q)= \mathrm{D}(M(X)\|M(Y)) \le \mathrm{D}(X\|Y).$$ The rest follows from Theorem~\ref{thm:privacy_technical_tradeoff}. It is well-known that Renyi divergence and KL divergence decrease with post-processing (see \cite{EH14}, for example). We can also compute $\mathrm{D}_\alpha(\cN(0,1),\cN(s,1))=\alpha s^2/2$ and $\mathrm{D}_{KL}(\cN(0,1),\cN(s,1))=s^2/2$~\cite{Mir17}.
\end{proof}

% Actually the bounds we get on Renyi divergence and Wasserstein distance are tight, where the tightness can be shown by Gaussian case.

\section{Efficient Non-smooth Sampling}\label{sec:sampling}
In this section, we will present an efficient sampling scheme for (non-smooth) functions to complement our main result first.
% Afterwards, as an example of the wide applications of our main result, we extend our result to differentially convex optimization, specifically DP-ERM and DP-SCO.
%\subsection{Sampling Scheme}
%\label{sec:sampling}s
Specifically, we study the following problem about sampling from a (non-smooth) log-concave distribution.
\begin{problem}\label{problem:sample}
Given a $\mu$-strongly convex function $\psi(x)$ defined on a convex set $\cK \subseteq \R^{d}$ and $+\infty$ outside. Given a family of $G$-Lipschitz convex functions $\{f_{i}(x)\}_{i\in I}$ defined on $\cK$. Our goal is to sample a point
$x\in \cK$ with probability proportionally to $\exp(-\widehat{F}(x))$ where
\[
\widehat{F}(x)=\E_{i\in I}f_{i}(x)+\psi(x).
\]
\end{problem}

Our sampler is based on the alternating sampling algorithm in \cite{LST21}
(See algorithm \ref{algo:AlternatingSampler}). This algorithm reduces
the problem of sampling from $\exp(-\widehat{F}(x))$ to sampling from $\exp(-\widehat{F}(x)-\frac{1}{2\eta}\|x-y\|^{2})$
for some fixed $\eta$ and for roughly $\frac{1}{\eta\mu}$ many different
$y$. When the step size $\eta$ is very small, the later problem
is easier because the distribution is almost like a Gaussian distribution.
For our problem, we will pick the largest step size $\eta$ such that
we can sample $\exp(-\widehat{F}(x)-\frac{1}{2\eta}\|x-y\|^{2})$ using only
$\widetilde{O}(1)$ many steps. 

\begin{algorithm2e}
\caption{Alternating Sampler \label{algo:AlternatingSampler}}

\textbf{Input:} $\mu$-strongly convex function $\widehat{F}$, step size $\eta>0$, initial
point $x_{0}$

\For{$t\in[T]$}{

$y_{t}\leftarrow x_{t-1}+\sqrt{\eta}\cdot\zeta$ where $\zeta\sim \cN(0, I_{d})$.

Sample $x_{t}\propto\exp(-\widehat{F}(x)-\frac{1}{2\eta}\|x-y_{t}\|_{2}^{2})$.\label{line:sample_oracle}

}

\textbf{Return} $x_{T}$

\end{algorithm2e}
\begin{theorem}[{\cite[Theorem 1]{LST21}}]
\label{thm:alternating_sampler}Given a $\mu$-strongly convex function
$F$ defined on $\cK$ with an initial point $x_{0}$. Let the distance $D=\|x_{0}-x^{*}\|_{2}$
for any $x^{*}=\arg\min_{x \in \cK}\widehat{F}(x)$. Suppose the step size $\eta\leq\frac{1}{\mu}$,
the target accuracy $\delta>0$ and the number of step $T\geq\Theta(\frac{1}{\eta\mu}\log(\frac{d/\mu+D^{2}}{\eta\delta}))$.
Then, Algorithm \ref{algo:AlternatingSampler} returns a random point
$x_{T}$ that has $\delta$ total variation distance to the distribution proportional to $\exp(-\widehat{F}(x))$.
\end{theorem}

Now, we show that Line \ref{line:sample_oracle} in Algorithm~\ref{algo:AlternatingSampler} can be implemented
by a simple rejection sampling. The idea is to pick step size $\eta$
small enough such that $\widehat{F}(x)$ is essentially a constant function
for a random $x\sim \cN(y,\eta\cdot I_{d})$. The precise algorithm
is given in Algorithm \ref{algo:AlternatingSamplerImpl}.

\begin{algorithm2e}

\caption{Implementation of Line \ref{line:sample_oracle} \label{algo:AlternatingSamplerImpl}}

\textbf{Input}: convex function $\widehat{F}(x)=\E_{i\in I}f_{i}(x)+ \psi(x)$,
step size $\eta>0$, current point $y$

\Repeat{$u\leq\frac{1}{2}\rho$}{

Sample $x,z$ from the distribution $\propto\exp(-\psi(x)-\frac{1}{2\eta}\|x-y\|^{2}_2)$

Set $\rho\leftarrow1$

\For{$\alpha=1,2,\cdots$}{

$\rho\leftarrow\rho+\Pi_{i=1}^{\alpha}(f_{j_{i}}(z)-f_{j_{i}}(x))$
where $j_{i}$ are random indices in $I$ 

With probability $\frac{\alpha}{1+\alpha}$, \textbf{break}

}

Sample $u$ uniformly from $[0,1]$.

}

\textbf{Return} $x$

\end{algorithm2e}

Since $F$ has the $\psi$ term, instead of sampling $x$ from $\cN(y,\eta\cdot I_{d})$,
we sample from $\exp(-\psi(x)-\frac{1}{2\eta}\|x-y\|^{2})$ in Algorithm
\ref{algo:AlternatingSamplerImpl}. The following lemma shows how
to decompose the distribution $\exp(-\widehat{F}(x)-\frac{1}{2\eta}\|x-y\|^{2})$
into the distribution mentioned above and the distribution $\exp(-\E_{i\in I}f_{i}(x))$. It
also calculates the distribution given by the algorithm. 
\begin{lem}
\label{lem:dpi}Let $\pi$ be the distribution proportional to $\exp(-\widehat{F}(x)-\frac{1}{2\eta}\|x-y\|_{2}^{2})$
and let $\mathcal{G}$ be the distribution proportional to $\exp(-\psi(x)-\frac{1}{2\eta}\|x-y\|^{2})$.
Then, we have that
\[
\frac{d\pi}{dx}=\frac{d\mathcal{G}}{dx}\cdot\frac{\exp(-\E_{i\in I}f_{i}(x))}{\E_{x\sim\mathcal{G}}\exp(-\E_{i\in I}f_{i}(x))}.
\]
Let $\widetilde{\pi}$ be the distribution returns by Algorithm \ref{algo:AlternatingSamplerImpl}.
Then, we have that
\[
\frac{d\widetilde{\pi}}{dx}=\frac{d\mathcal{G}}{dx}\cdot\frac{\E(\overline{\rho}|x)}{\E(\overline{\rho})}
\]
where $\overline{\rho} = \min(\max(\rho,0),2)$ is the truncation of $\rho$ in
Algorithm \ref{algo:AlternatingSamplerImpl} to $[0,2]$, $\E(\overline{\rho}|x)$ is the expected value of $\overline{\rho}$ conditional on $x$, and
$\E(\overline{\rho} )=\E_{x\sim\mathcal{G}}\E(\overline{\rho}|x)$. Furthermore,
we have that
\[
\E(\rho|x)=\exp(-\E_{i\in I}f_{i}(x))\cdot\E_{z\sim\mathcal{G}}\exp(\E_{i\in I}f_{i}(z)).
\]
\end{lem}

\begin{proof}
For the true distribution $\pi$, we have
\begin{align*}
\frac{d\pi}{dx} & =\frac{\exp(-\E_{i\in I}f_{i}(x)-\psi(x)-\frac{1}{2\eta}\|x-y\|_{2}^{2})}{\int\exp(-\E_{i\in I}f_{i}(x)-\psi(x)-\frac{1}{2\eta}\|x-y\|_{2}^{2})dx} \\
& = \frac{\exp(-\E_{i\in I}f_{i}(x))\frac{d\mathcal{G}}{dx}}{\int\exp(-\E_{i\in I}f_{i}(x))\frac{d\mathcal{G}}{dx}dx}=\frac{d\mathcal{G}}{dx}\cdot\frac{\exp(-\E_{i\in I}f_{i}(x))}{\E_{x\sim\mathcal{G}}\exp(-\E_{i\in I}f_{i}(x))}.
\end{align*}

For the distribution $\widetilde{\pi}$ by the algorithm, we sample
$x\sim\mathcal{G}$, then accept the sample if $u\leq\frac{1}{2}\rho$.
Hence, we have
\[
\frac{d\widetilde{\pi}}{dx}=\frac{d\mathcal{G}}{dx}\frac{\Pr(u\leq\frac{1}{2}\rho|x)}{\Pr(u\leq\frac{1}{2}\rho)}.
\]
Since $u$ is uniform between $0$ and $1$, we have the result.

Finally, for the expectation of $\rho$, we note that 
\[
\E\Pi_{i=1}^{\alpha}(f_{j_{i}}(z)-f_{j_{i}}(x))=(\E_{i\in I}(f_{i}(z)-f_{i}(x)))^{\alpha}
\]
and that the probability that the loop pass step $\alpha$ is exactly
$\frac{1}{\alpha!}$. Hence, we have 
\[
\E(\rho|x,z)=1+\sum_{\alpha=1}^{\infty}\frac{1}{\alpha!}(\E_{i\in I}(f_{i}(z)-f_{i}(x)))^{\alpha}=\exp(\E_{i\in I}(f_{i}(z)-f_{i}(x)).
\]
Taking expectation over $z$ gives the result.
\end{proof}
Note that if we always had $0 \leq \rho\leq2$, then $\E(\overline{\rho}|x)=\E(\rho|x)\propto\exp(-\E_{i\in I}f_{i}(x))$
and hence $\frac{d\pi}{dx}=\frac{d\widetilde{\pi}}{dx}$. Therefore,
the only thing left is to show that $0\leq \rho\leq2$ with high probability
and that it does not induces too much error in total variation distance.
To do this, we use Gaussian concentration to prove that $\E_{i\in I}f_{i}(x)$
is almost a constant over random $x\sim\mathcal{G}$.
\begin{lem}[{Gaussian concentration \cite[Eq 1.21]{Led99}}]
\label{lem:gaussian_concentration}Let $X \sim \exp(-\widehat{F})$ for some $1/\eta$-strongly convex $\widehat{F}$ and $\ell$
is a $G$-Lipschitz function. Then, for all $t\geq0$,
\[
\Pr[\ell(X)-\mathbb{E}[\ell(X)]\geq t]\leq e^{-t^{2}/(2\eta G^2)}.
\]
\end{lem}

Now, we are already to prove our main result. This shows that if $\eta\ll G^{-2}$,
then the algorithm indeed implements Line \ref{line:sample_oracle}
correctly up to small error.
\begin{lem}
\label{lem:sub_problem}If the step size $\eta\leq C\log^{-1}(1/\delta_{\mathrm{inner}})G^{-2}$
for some small enough $C$ and the inner accuracy $\delta_{\mathrm{inner}}\in(0,1/2)$,
then Algorithm \ref{algo:AlternatingSamplerImpl} returns a random
point $x$ that has $\delta_{\mathrm{inner}}$ total variation distance
to the distribution proportional to $\exp(-\widehat{F}(x)-\frac{1}{2\eta}\|x-y\|_{2}^{2})$. Furthermore,
the algorithm accesses only $O(1)$ many $f_{i}(x)$ in expectation and samples from $\exp(-\psi(x) - \frac{1}{2\eta} \|x-y\|^2_2)$ for $O(1)$ many $y$.
\end{lem}

\begin{proof}
Let $\pi$ be the distribution given by $c\cdot\exp(-\widehat{F}(x)-\frac{1}{2\eta}\|x-y\|_{2}^{2})$
and $\widetilde{\pi}$ is the distribution outputted by the algorithm.
By Lemma \ref{lem:dpi}, we have
\begin{align*}
d_{\mathrm{TV}}(\pi,\widetilde{\pi}) & =\int_{\R^{d}}\left|\frac{d\mathcal{G}}{dx}\frac{\exp(-\E_{i\in I}f_{i}(x))}{\E_{x\sim\mathcal{G}}\exp(-\E_{i\in I}f_{i}(x))}-\frac{d\mathcal{G}}{dx}\frac{\E(\overline{\rho}|x)}{\E(\overline{\rho})}\right|dx\\
 & =\E_{x\sim\mathcal{G}}\left|\frac{\exp(-\E_{i\in I}f_{i}(x))}{\E_{x\sim\mathcal{G}}\exp(-\E_{i\in I}f_{i}(x))}-\frac{\E(\overline{\rho}|x)}{\E(\overline{\rho})}\right|.
\end{align*}
Let $X$ be the random variable $\E(\rho|x)$ and $\widetilde{X}$
be the random variable $\E(\overline{\rho}|x)$. Lemma \ref{lem:dpi}
shows that $X=\exp(-\E_{i\in I}f_{i}(x))\cdot\E_{z\sim\mathcal{G}}\exp(\E_{i\in I}f_{i}(z))$
and hence
\[
\frac{\exp(-\E_{i\in I}f_{i}(x))}{\E_{x\sim\mathcal{G}}\exp(-\E_{i\in I}f_{i}(x))}=\frac{X}{\E_{x\sim \G}X}.
\]
Therefore, we have
\begin{align}
d_{\mathrm{TV}}(\pi,\widetilde{\pi}) & =\E\left|\frac{X}{\E X}-\frac{\widetilde{X}}{\E\widetilde{X}}\right|\leq\E\left|\frac{X}{\E X}-\frac{\widetilde{X}}{\E X}\right|+\E\left|\frac{\widetilde{X}}{\E X}-\frac{\widetilde{X}}{\E\widetilde{X}}\right|\leq2\frac{\E|X-\widetilde{X}|}{|\E X|}.\label{eq:dTV1}
\end{align}

We simplify the right hand side by lower bounding $\E X$. By Lemma
\ref{lem:gaussian_concentration} and the fact that the negative log-density of $\mathcal{G}$ is $1/\eta$-strongly convex, we have that $\E_{i\in I}f_{i}(z)\geq\E_{x\sim\mathcal{G}}\E_{i\in I}f_{i}(x)-2G \sqrt{\eta}$
with probability $\geq1-e^{-2}$. Hence, we have
\begin{align*}
\E X & =\E_{x\sim\mathcal{G}}\exp(-\E_{i\in I}f_{i}(x))\cdot\E_{z\sim\mathcal{G}}\exp(\E_{i\in I}f_{i}(z))\\
 & \geq\exp(-\E_{x\sim\mathcal{G}}\E_{i\in I}f_{i}(x))\cdot\E_{z\sim\mathcal{G}}\exp(\E_{i\in I}f_{i}(z))\\
 & =\E_{z\sim\mathcal{G}}\exp(\E_{i\in I}f_{i}(z)-\E_{x\sim\mathcal{G}}\E_{i\in I}f_{i}(x))\\
 & \geq(1-e^{-2})\exp(-2G \sqrt{\eta}).
\end{align*}
Using $\eta\leq G^{-2}/8$,
we have $\E [X]\geq\frac{2}{3}$. Using this, (\ref{eq:dTV1}), $X=\E(\rho|x)$
and $\widetilde{X}=\E(\overline{\rho}|x)$, we have
\[
d_{\mathrm{TV}}(\pi,\widetilde{\pi})\leq 3 \cdot\E|X-\widetilde{X}|\leq 3\cdot\E(|\rho|\cdot1_{\rho \notin [0,2]}).
\]

% \Daogao{...}
% Now we bound the expectation $\E(|\rho|\cdot1_{\rho \notin [0,2]})$.
% By a calculation similar to Lemma \ref{lem:dpi}, we have
% \begin{align*}
%     \E(|\rho|\cdot1_{\rho \notin [0,2]})\le&~\E(|\rho|)\\
%     \le & ~ \E_{x,z}\exp(\E_{i\in I}|f_i(z)-f_i(x)|).
% \end{align*}

% Thus
% \begin{align*}
%     \E_{x,z}\exp(\E_{i\in I}|f_i(z)-f_i(x)|)\le& \int_{0}^{\infty} \Pr_{x,y}[\exp(\E_{i\in I}|f_i(z)-f_i(x)|)\ge t] \d t\\
%     \leq & 1+\int_{1}^{\infty} 6e^{-\log^2 t/(64\eta G^2)}\d t\\
%     \leq& 1+ \int_{1}^{\infty}6e^{}
% \end{align*}

% \Daogao{.....Expectation too large. Can we bound the tail probability instead?}

We split the $\rho$ into two terms $\rho_{\leq L}$ and $\rho_{>L}$.
The first term $\rho_{\leq L}$ is the sum of all terms added to $\rho$
when $\alpha\leq L$ (including the initial term $1$). The second
term $\rho_{>L}$ is the sum when $\alpha>L$. Hence, we have $\rho=\rho_{>L}+\rho_{\leq L}$
and hence
\begin{equation}
d_{\mathrm{TV}}(\pi,\widetilde{\pi})\leq3\cdot\E(|\rho_{>L}|\cdot1_{\rho \notin [0,2]})+3\cdot\E(|\rho_{\leq L}|\cdot1_{\rho \notin [0,2]}).\label{eq:dTV2}
\end{equation}

For the term $\rho_{>L}$, by a calculation similar to Lemma \ref{lem:dpi},
we have
\begin{align*}
\E(|\rho_{>L}|\cdot1_{\rho \notin [0,2]}) & \leq\E|\rho_{>L}|\leq\E_{x,z}\Phi(\E_{i\in I}|f_{i}(z)-f_{i}(x)|),
\end{align*}
% \Daogao{Actually, the bound seems to be $\E|\rho_{>L}|\leq\E_{x,z}\Phi(\E_{i\in I}|f_{i}(z)-f_{i}(x)|).$ }
where $\Phi(t)=\sum_{\alpha=L+1}^{\infty}\frac{t^{\alpha}}{\alpha!}$ is a power series in $t$ with all positive coefficients.
By picking $L>C\log(1/\delta_{\mathrm{inner}})$ for some large constant
$C$, we have $\Phi(t)\leq\frac{\delta_{\mathrm{inner}}}{16}$
for all $|t|\leq1$. Let $\Delta$ be the random variable $\E_{i\in I}|f_{i}(z)-f_{i}(x)|$ whose randomness comes from $x$ and $z$.
Then, we have
\[
\E(|\rho_{>L}|\cdot1_{\rho \notin [0,2]})\leq\frac{\delta_{\mathrm{inner}}}{16}+\E e^{\Delta}1_{\Delta\geq1}\leq\frac{\delta_{\mathrm{inner}}}{16}+\sum_{k=1}^{\infty}e^{k+1}\Pr_{x,z}(\Delta\geq k).
\]

Denote a function $h_{x,z}(t):=\Pr_{i\in I}[|f_i(z)-f_i(x)|\ge t]$.
Since each $f_i$ is $G$-Lipschitz, Lemma~\ref{lem:gaussian_concentration} shows that
\begin{align*}
    \Pr_{x,z}[|f_i(z)-f_i(x)|\geq t]\leq 4e^{-t^2/(8\eta G^2)},
\end{align*}
which implies 
\begin{align*}
    \E_{x,z}[h_{x,z}(t)] = \Pr_{x,z,i}[|f_i(z)-f_i(x)|\geq t]\leq 4e^{-t^2/(8\eta G^2)}.
\end{align*}
By Markov inequality, for any $k>0$, we know
\begin{align*}
    \Pr_{x,z}[h_{x,z}(t)\geq e^{-k}]\leq 4e^{k-t^2/(8\eta G^2)}.
\end{align*}
As $|f_i(z)-f_i(x)|\leq G\|x-z\|_2$, if $h_{x,z}(t)=\Pr_{i\in I}[|f_i(z)-f_i(x)|\ge t]\leq e^{-t^2/(16\eta G^2)}$, we know $$\E_{i\in I}|f_i(z)-f_i(x)|\leq t +e^{-t^2/(16\eta G^2)}\cdot G\|x-z\|_2.$$
Hence, one has
\begin{align*}
     \Pr_{x,z}\Big[\E_{i\in I}|f_i(z)-f_i(x)|\ge t+e^{-t^2/(16\eta G^2)}G\|x-z\|_2 \Big]
    \leq &~ \Pr_{x,z}[h_{x,z}(t)\ge e^{-t^2/(16\eta G^2)}]\\
    \leq & ~ 4e^{-t^2/(16\eta G^2)}.
\end{align*}
By Gaussian Concentration, we know
\begin{align*}
    \Pr_{x,z}[\|x-z\|_2\ge t]
    \leq &~ \Pr_{x,z}[\|x-\E x\|_2\ge t/2 \text{ or }\|z-\E z\|\ge t/2]\\
    \leq & ~ 2e^{-t^2/(8\eta)}.
\end{align*}
Thus we know
\begin{align*}
    &~ \Pr_{x,z}[\E_{i\in I}|f_i(z)-f_i(x)|\ge 2t]\\
    =& ~\Pr_{x,z}[\E_{i\in I}|f_i(z)-f_i(x)|\ge 2t, \|x-z\|_2\ge t/G]+\Pr_{x,z}[\E_{i\in I}|f_i(z)-f_i(x)|\geq 2t,\|x-z\|_2<t/G]\\
    \leq & ~2e^{-t^2/(8G^2\eta)}+\Pr_{x,z}[\E_{i\in I}|f_i(z)-f_i(x)|\geq 2t,\|x-z\|_2<t/G]\\
    \leq &~  2e^{-t^2/(8G^2\eta)} +\Pr_{x,z}[\E_{i\in I}|f_i(z)-f_i(x)|\ge t+e^{-t^2/(16\eta G^2)}G\|x-z\|_2]\\
    \leq & ~ 6e^{-t^2/(16\eta G^2)}.
\end{align*}
% \Daogao{new :}We want to bound the tail probability of $\Delta$. We can lose a $\log(|I|)$ term for union bound, or lose $\sqrt{d}$ to bound the distance between $\|x-z\|$.
% Now we try to get the original result for infinitely large size of $I$.
% \Daogao{Old:}
% Since $\E_{i\in I}f_{i}(x)$ is a $G$-Lipschitz function, Lemma \ref{lem:gaussian_concentration}
% shows that
% \[
% \Pr_{x,z}[|\E_{i\in I}(f_{i}(z)-f_{i}(x))|\geq t]\leq4e^{-t^{2}/(8 \eta G^2)}.
% \]
Hence, we have $\Pr(\Delta\geq k)\leq 6\exp(-k^{2}/(64G^{2}\eta))$ and 
\begin{equation}
\E(|\rho_{>L}|\cdot1_{\rho \notin [0,2]})\leq\frac{\delta_{\mathrm{inner}}}{16}+17\sum_{k=1}^{\infty}e^{k-\frac{k^{2}}{64G^{2}\eta}} \leq\frac{\delta_{\mathrm{inner}}}{9},\label{eq:dTv3}
\end{equation}
where we used $\eta\leq2^{-6}G^{-2}/\log(400/\delta_{\mathrm{inner}})$ at the end.

As for the term $\rho_{\leq L}$, we know that
\begin{align}
\label{eq:split_dTVterm}
& \E(|\rho_{\leq L}|\cdot 1_{\rho \notin [0,2]})\nonumber\\
= & \E(|\rho_{\leq L}|\cdot 1_{\rho \notin [0,2]}\cdot 1_{|\rho_{\leq L}|\leq 2^L})+\E(|\rho_{\leq L}|\cdot1_{\rho \notin [0,2]}\cdot 1_{|\rho_{\leq L}|\geq 2^L})\nonumber \\
\leq & \Pr[\rho\notin[0,2]]\cdot 2^L+ \sum_{k=1}^{\infty}2^{(k+1)L}\Pr(|\rho_{\leq L}|\ge2^{kL}).
\end{align}
Note that the term $\rho_{\leq L}$ involves only less than $\frac{L^{2}}{2}$
many $f_{i}(x)$ and $f_{i}(z)$. Lemma \ref{lem:gaussian_concentration}
shows that for any $i$, we have
\[
\Pr_{x\sim\mathcal{G}}(|f_{i}(x)-\E_{x\sim\mathcal{G}}f_{i}(x)|\geq t )\leq2e^{-t^{2}/(2 \eta G^2)}.
\]
By union bound, this shows 
\[
\Pr_{x,z\sim\mathcal{G}}(|f_{i}(x)-f_{i}(z)|\geq\frac{1}{4}2^{k}\text{ for any such }i)\leq L^{2}\exp(-\frac{4^{k}}{32\eta G^{2}}).
\]
Under the event $|f_{i}(x)-f_{i}(z)|\leq\frac{1}{3}2^{k}$
for all $i$ appears in $\rho_{\leq L}$, we have
\[
|\rho_{\leq L}|\leq1+\sum_{\alpha=1}^{L}\Pi_{i=1}^{\alpha}|f_{j_{i,\alpha}}(z)-f_{j_{i,\alpha}}(x)|\leq1+\sum_{\alpha=1}^{L}(\frac{2^{k}}{3})^\alpha\leq2^{kL}.
\]
Therefore, we have $\Pr(|\rho_{\leq L}|>2^{kL})\leq L^{2}\exp(-\frac{4^{k}}{32\eta G^{2}})$
and
\[
\sum_{k=1}^{\infty}2^{(k+1)L}\Pr(|\rho_{\leq L}|>2^{kL})\leq\sum_{k=1}^{\infty}2^{(k+1)L}L^{2}\exp(-\frac{4^{k}}{32\eta G^{2}})\leq\sum_{k=1}^{\infty}2^{4kL}\exp(-\frac{4^{k}}{32\eta G^{2}}).
\]
Picking $\eta\leq2^{-8}G^{-2}L^{-1}$, we have that
\begin{equation}
\sum_{k=1}^{\infty}2^{(k+1)L}\Pr(|\rho_{\leq L}|>2^{kL})\leq\sum_{k=1}^{\infty}2^{4kL}\exp(-2\cdot4^{k}L)\leq\sum_{k=1}^{\infty}2^{-kL}\leq\frac{\delta_{\mathrm{inner}}}{9}\label{eq:dTV4}
\end{equation}
by picking $L>C\log(1/\delta_{\mathrm{inner}})$ for large enough
$C$.

It remains to bound the term $\Pr[\rho\notin[0,2]]\cdot 2^L$.
We know the probability the algorithm enters the $(L+1)$-th phase is at most $\frac{1}{L!}\leq \frac{2}{2^L}$.
Hence we know $\Pr[\rho \notin[0,2]]\leq \frac{2}{2^L}+\Pr[\rho_{\leq L}\notin[0,2]]$.
Similarly, by Gaussian Concentration and union bound, we have
\begin{align*}
    \Pr_{x,z\sim\mathcal{G}}(|f_{i}(x)-f_{i}(z)|\geq1/2\text{ for any such }i)\leq L^{2}\exp(-\frac{1}{8\eta G^{2}}).
\end{align*}
Under the event that $|f_i(x)-f_i(z)|\leq 1/2$ for all $i$ appears in $\rho_{\leq L}$, we have
\begin{align*}
    1-\sum_{\alpha=1}^{L}\Pi_{i=1}^{\alpha}|f_{j_{i,\alpha}}(z)-f_{j_{i,\alpha}}(x)|\leq \rho_{\le L}\le 1+\sum_{\alpha=1}^{L}\Pi_{i=1}^{\alpha}|f_{j_{i,\alpha}}(z)-f_{j_{i,\alpha}}(x)|,
\end{align*}
which implies $0\le \rho_{\leq L}\leq 2$.
Then we know $\Pr[\rho_{\le L}\notin[0,2]]\leq L^2\exp(-\frac{1}{8\eta G^2})$.
By our setting of parameters and that $L = C \log(1/\delta_{\mathrm{inner}})$ for some large constant $C$, we know
\begin{align}
\label{eq:dTV5}
    \Pr[\rho\notin[0,2]]\cdot 2^L\leq 2^L(L^2\exp(-\frac{1}{8\eta G^2})+\frac{2}{2^L})\leq \frac{\delta_{\mathrm{inner}}}{9}.
\end{align}

Combining (\ref{eq:dTV2}), (\ref{eq:dTv3}), (\ref{eq:split_dTVterm}), (\ref{eq:dTV4}) and (\ref{eq:dTV5}),
we have the result $d_{\mathrm{TV}}(\pi,\widetilde{\pi})\leq\delta_{\mathrm{inner}}$.

Finally, the accept probability is given by $\E\widetilde{X}/2$ and $\E\widetilde{X}\geq\E X-\E|X-\widetilde{X}|\geq\frac{2}{3}-\frac{\delta_{\mathrm{inner}}}{3}\geq\frac{1}{3}$.
Hence, the number of access is $O(1)$.
\end{proof}
Combining Theorem \ref{thm:alternating_sampler} and Lemma \ref{lem:sub_problem},
we have the following result:
\begin{theorem}
\label{thm:sampler}
Given a $\mu$-strongly convex function $\psi(x)$ defined on a convex set $\cK \subseteq \R^{d}$ and $+\infty$ outside. Given a family of $G$-Lipschitz convex functions $\{f_{i}(x)\}_{i\in I}$ defined on $\cK$.
Define the function $\widehat{F}(x)=\E_{i\in I}f_{i}(x)+\psi(x)$ and 
the distance $D=\|x_{0}-x^{*}\|_{2}$ for some $x^{*}=\arg\min_{x}\widehat{F}(x)$.
For any $\delta\in(0,1/2)$, if we can get samples from $\exp(-\psi(x)-\frac{\|x-y\|_2^2}{2\eta}) $ for any $y\in\R^d$ and $\eta>0$, we can find a random point $x$ that
has $\delta$ total variation distance to the distribution proportional to $\exp(-\widehat{F}(x))$ in
\[
T:=\Theta(\frac{G^{2}}{\mu}\log^{2}(\frac{G^{2}(d/\mu+D^{2})}{\delta}))\text{ steps}.
\]
Furthermore, each steps accesses only $O(1)$ many $f_{i}(x)$ in
expectation and samples from $\exp(-\psi(x) - \frac{1}{2\eta} \|x-y\|^2_2)$ for $O(1)$ many $y$ with $\eta = \Theta(G^{-2}/\log(T/\delta))$.
\end{theorem}

\begin{proof}
This follows from applying Lemma \ref{lem:sub_problem} to implement
Line \ref{line:sample_oracle}. Note that the distribution implemented
has total variation distance $\delta_{\mathrm{inner}}$ to the required
one. By setting $\delta_{\mathrm{inner}}=\delta/(2T)$, this only
gives an extra $\delta/2$ error in total variation distance. Finally,
setting $\eta=\Theta(G^{-2}/\log(1/\delta_{\mathrm{inner}}))$, Theorem
\ref{thm:alternating_sampler} shows that Algorithm \ref{algo:AlternatingSamplerImpl}
outputs the correct distribution up to $\delta/2$ error in total
variation distance. This gives the result.
\end{proof}

In the most important case of interest when $\psi(x)$ is $\ell_2^2$ regularizer, one can see $\exp(-\psi(x) - \frac{1}{2\eta} \|x-y\|^2_2)$ is a truncated Gaussian distribution, and there are many results on how to sample from truncated Gaussian, e.g. \cite{KD99}.
For more general case, there are also efficient algorithms to do the sampling, such as the Projected Langevin Monte Carlo \cite{BEL18}.
In fact our sampling scheme matches the information-theoretical lower bound on the value query complexity up to some logarithmic terms, which can be reduced from the result in \cite{DJWW15} with some modifications.
See Section~\ref{sec:infolower} for a detailed discussion.

\section{DP Convex Optimization}
In this section we present our results about DP-ERM and DP-SCO.

\subsection{DP-ERM}
In this subsection, we state our result for the DP-ERM problem \eqref{eq:DPERM}.
Briefly speaking, our main result (Theorem~\ref{thm:privacy_technical}) shows that sampling from $\exp(-kF(x;\cD))$ for some appropriately chosen $k$ is $(\eps,\delta)$-DP and achieves the optimal empirical risk in (\ref{eqn:optimal_empirical_risk}).
Our sampling scheme in Section~\ref{sec:sampling} provides an efficient implementation. We start with the following lemma which shows the utility guarantee for the sampling mechanism.
%Here, we give a proof for the convex setting with the tight constant. The proof is similar to \cite{KV06,BST14}.

\begin{lemma}[Utility Guarantee, {\cite[Corollary 1]{DKL18}}]
\label{lm:utility_tech}
Suppose $k>0$ and $F$ is a convex function over the convex set $\cK\subseteq \R^d$. If we sample $x$ according to distribution $\nu$ whose density is proportional to $\exp(-k F(x))$, then we have
\begin{align*}
    \E_{\nu}[F(x)]\leq \min_{x\in\cK}F(x)+\frac{d}{k}.
\end{align*}
\end{lemma}

This is first shown by \cite{KV06} for any linear function $F$, and \cite{BST14} extends it to any convex function $F$ with a slightly worse constant. 

\begin{theorem}[DP-ERM]\label{thm:DPERM}
Let $\epsilon>0$, $\cK\subseteq \R^d$ be a convex set of diameter $D$ and $\{f(\cdot;s)\}_{s\in\mathcal{D}}$ be a family of convex functions over $\cK$ such that $f(x;s)-f(x;s')$ is $G$-Lipschitz for all $s,s'$.
For any data-set $\cD$ and $k>0$, sampling $x^{(priv)}$ with probability proportional to $\exp\left(-k(F(x;\cD)+\mu\|x\|_2^2/2)\right)$ is $(\epsilon,\delta(\eps))$-differentially private, where
\begin{align*}
 \delta(\eps)\leq \deltacurve{\cN(0,1)}{\cN\lp\frac{G\sqrt{k}}{n\sqrt{\mu}},1\rp}(\epsilon).
\end{align*}
The excess empirical risk is bounded by {$\frac{d}{k}+\frac{\mu D^2}{2}$}.
Moreover, if $\{f(\cdot,s)\}_{s\in\mathcal{D}}$ are already $\mu$-strongly convex, then sampling 
$x^{(priv)}$ with probability proportional to $\exp(-kF(x;\cD))$ is $(\eps,\delta(\eps))$-differentially private where 
\begin{align*}
 \delta(\eps)\leq \deltacurve{\cN(0,1)}{\cN\lp\frac{G\sqrt{k}}{n\sqrt{\mu}},1\rp}(\epsilon).
 \end{align*}
 The excess empirical risk is bounded by $\frac{d}{k}$.
\end{theorem}
% \Gopi{Do we have the factor 2 in the privacy bound?}
\begin{proof}
The privacy guarantee follows directly from our main result Theorem~\ref{thm:privacy_technical}, and the bound on excess empirical loss can be proved by Lemma~\ref{lm:utility_tech}.
\end{proof}
%\begin{remark}
%If we use the result about Gaussian mechanism, {\cite[Theorem 3.22]{DR14}}, we can set $k=\frac{\eps^2n^2\mu}{8G^2\log(1.25/\delta)}$ in the strongly convex case with expected excess empirical loss $\frac{8G^2d\log(1.25/\delta)}{\eps^2n^2\mu} $. In the general convex case, we can set $k=\frac{\eps^2n^2\mu}{8G^2\log(1.25/\delta)}$ where $\mu=\frac{4G\sqrt{d\log(1.25/\delta)}}{\eps n D}$ and get expected excess empirical loss $\frac{4GD\sqrt{d\log(1.25/\delta)}}{\eps n}$.
%\end{remark}

%\Gopi{The bound from {\cite[Theorem 3.22]{DR14}} is only stated in the regime $\eps\in (0,1).$ Find an other reference or use the bounds from Section 3.}
%Tat  I think just stating for eps between 0,1 is okay just for statement simplicity

Before we state the implementation results on DP-ERM, we need the following technical lemma:
\begin{lemma}
\label{lm:privacy_curve_bound}
For any constants $1/2>\delta>0$ and $\eps>0$, if $|s|\leq \sqrt{2\log(1/(2\delta))+2\eps}-\sqrt{2\log(1/(2\delta))}$,
one has
\begin{align*}
  \deltacurve{\cN(0,1)}{\cN(s,1)}\leq \delta.  
\end{align*}
\end{lemma}
\begin{proof}
By Equation~(\ref{eqn:Gaussian_privacycurve}), we know that
\begin{align*}
    \deltacurve{\cN(0,1)}{\cN(s,1)}(\eps)\leq \Phi\lp -\frac{\eps}{s}+\frac{s}{2} \rp.
\end{align*}
Without loss of generality, we assume $s\geq0$ and
want to find an appropriate value of $s$ such that $\Phi\lp -\frac{\eps}{s}+\frac{s}{2} \rp\leq \delta$. 
Denote $t\defeq \Phi^{-1}(1-\delta)$ and since $1-\Phi(t)\le \frac{1}{2}\exp(-t^2/2)$ for $t>0$, we know that $t\leq \sqrt{2\log(1/(2\delta))}$.
It is equivalent to solve the equation $\frac{\eps}{s}-\frac{s}{2}\geq t$, which is equivalent to $0\leq s\leq \sqrt{t^2+2\eps}-t$.
Note that $\sqrt{t^2+2\eps}-t$ decreases as $t$ increases, which implies that we can set $s\leq \sqrt{2\log(1/(2\delta))+2\eps}-\sqrt{2\log(1/(2\delta))}$.
\end{proof}

Combining the sampling scheme (Theorem~\ref{thm:sampler}) and our analysis on DP-ERM, we can get the efficient implementation results on DP-ERM directly.

\begin{theorem}[DP-ERM Implementation]
\label{thm:DPERM_impl}
With same assumptions in Theorem~\ref{thm:DPERM}, and assume $f(\cdot;s)$ is $G$-Lipschitz over $\cK$ for all $s$.
For any constants $1/10> \delta>0$ and $ \eps> 0$, there is an efficient sampler to solve DP-ERM which has the following guarantees:
\begin{itemize}
    \item The scheme is $(\eps,\delta)$-differentially private;
    \item The expected excess empirical loss is bounded by $\frac{GD\sqrt{d}}{n(\sqrt{\log(1/\delta)+\eps}-\sqrt{\log(1/\delta)})}$.
    In particular, if $\eps< 1/10$, the expected excess empirical loss is bounded by
    $
        \frac{2GD\sqrt{d\log(1/\delta)}}{\eps n}.
    $
    If $\eps \geq \log(1/\delta)$, the expected excess empirical loss is bounded by $ O(\frac{GD\sqrt{d}}{n\sqrt{\eps}})$.
    \item The scheme takes 
    \begin{align*}
        \Theta\lp\frac{\eps^2n^2}{\log(1/\delta)}\log^2(\frac{nd\eps}{\delta})\rp
    \end{align*}
    queries to the values on $f(x;s)$ in expectation and takes the same number of samples from some Gaussian restricted to the convex set $\mathcal{K}$.
\end{itemize}
\end{theorem}

\begin{proof}
% {\cite[Theorem 3.22]{DR14}} shows that $\delta(N(0,1)||N(\Delta,1))(\epsilon)\leq\delta$
% if 
% \[
% \frac{\sqrt{2\ln(1.25/\delta)}\Delta}{\epsilon}\leq1.
% \]

By Lemma~\ref{lm:privacy_curve_bound}, we can set $s=\sqrt{2\log(3/(4\delta))+2\eps}-\sqrt{2\log(3/(4\delta))}$ to make $\deltacurve{\cN(0,1)}{\cN(s,1)}\leq2\delta/3$.
For our setting, Theorem \ref{thm:DPERM} shows that we have $s=\frac{G\sqrt{k}}{n\sqrt{\mu}}$
and hence we can take
\begin{align*}
    k=\frac{2\mu n^{2}\lp\sqrt{\log(3/(4\delta))+\eps}-\sqrt{\log(3/(4\delta))}\rp^2}{G^{2}}.
\end{align*}
Putting it into the excess empirical loss bound of $\frac{d}{k}+\frac{\mu D^{2}}{2}$
and setting $\mu=\frac{G\sqrt{d}}{ n D\lp\sqrt{\log(3/(4\delta))+\eps}-\sqrt{\log(3/(4\delta))}\rp}$,
we get the result on the empirical loss.

Particularly, consider the case when $\eps<1/10$.
We know the excess empirical loss is bounded by $\frac{GD\sqrt{d}}{n(\sqrt{\log(3/(4\delta))+\eps}-\sqrt{\log(3/(4\delta))})}$.
Note that $1+\frac{x}{2}-\frac{x^2}{8}\leq \sqrt{1+x}\leq 1+\frac{x}{2}$ for $x\geq 0$.
Under the assumption that $\delta,\eps\in(0,\frac{1}{10})$, we know $\frac{GD\sqrt{d}}{n(\sqrt{\log(3/(4\delta))+\eps}-\sqrt{\log(3/(4\delta))})}\leq \frac{2GD\sqrt{d\log(4/(5\delta))}}{n\eps}$.
The case when $\eps\geq \log(1/\delta)$ also follows similarly.
% Particularly, if $\eps\leq \log(3/(4\delta))$, we know $0\leq s\leq \frac{\eps}{\sqrt{2\log(1/\delta)}}\leq \sqrt{2\log(3/(4\delta))+2\eps}-\sqrt{2\log(3/(4\delta))}\leq \sqrt{t^2+2\eps}-t$.
% Hence we can set $k$.
% \gnote{Can remove the text from here}
% Particularly, consider the case when $\eps<1/2$.
% By Equation~\eqref{eqn:Gaussian_privacycurve_approx}, we know that 
% \begin{align*}
%     \deltacurve{\cN(0,1)}{\cN(s,1)}(\eps)  \le \frac{1}{2} \exp\lp-\frac{1}{2}\lp\frac{\eps}{s}-\frac{s}{2}\rp^2\rp  \text{ if } \eps \geq \frac{s^2}{2}.
% \end{align*}
% Hence we know $\deltacurve{\cN(0,1)}{\cN(s,1)}(\eps)\leq 4\delta/5$ if
% $
%     \frac{s\sqrt{2\log(3/(4\delta))}}{\eps}\leq 1
% $
% under the assumptions that $\eps< 1/2$ and $\delta< 1/2$.
% For our setting, Theorem \ref{thm:DPERM} shows that we have $s=\frac{2G\sqrt{k}}{n\sqrt{\mu}}$
% and hence we can take
% \[
% k=\frac{\epsilon^{2}n^{2}\mu}{8G^{2}\log(3/(4\delta))}.
% \]
% Putting it into the excess empirical loss bound of $\frac{d}{k}+\frac{\mu D^{2}}{2}$
% and setting $\mu=\frac{4G\sqrt{d\log(3/(4\delta))}}{\epsilon n D}$,
% we get the result on the empirical loss.
% The case when $\eps\geq \log(1/\delta)$ follows directly by the calculation.
% \gnote{to here. Instead just use upper bounds on the general bound just derived to get the special cases.}

To make it algorithmic, we apply Theorem~\ref{thm:sampler} with the accuracy on the total variation distance to be $\min\{\delta/3,\frac{1}{cn^c\eps}\}$ for some large enough constant $c$. This leads to $(\epsilon,\delta)$-DP and an extra empirical loss and hence we use $\log(1/\delta)$ rather than $\log(3/(4\delta))$ or $\log(4/(5\delta))$ in the final loss term. 

%Note that $1+\frac{x}{2}-\frac{x^2}{8}\leq \sqrt{1+x}\leq 1+\frac{x}{2}$ for $x\geq 0$. \gnote{Why is this fact mentioned here?}
The running time follows from Theorem~\ref{thm:sampler}.
\end{proof}

\subsection{DP-SCO and Generalization Error}
As mentioned before, one can reduce the DP-SCO \eqref{eq:DPSCO} to DP-ERM \eqref{eq:DPERM} by the iterative localization technique proposed by \cite{FKT20}.
But this method forces us to design different algorithms for DP-ERM and DP-SCO, and may lead to a large constant in the final loss.
In this section, we show that the exponential mechanism can achieve both the optimal empirical risk for DP-ERM and the optimal population loss for DP-SCO by simply changing the parameters.
The bound on the generalization error works beyond differential privacy and can be useful for other (non-private) optimization settings.

% We define the log-Sobolev inequality first.
% \begin{definition}[Log-Sobolev Inequality]
% We say a distribution $\nu$ satisfies a logarithmic Sobolev inequality with a constant $C$ if for all smooth
% function $g:\R^n\rightarrow \R$ with $\E_{\nu}[g^2]<\infty$, one has
% \begin{equation}
%     \label{eq:log-sobolev}
%      \Ent_\nu[g^2] \leq 2C \mathbb{E}_{\nu}\left[\|\nabla g\|_2^{2}\right]
% \end{equation}
% where $\Ent_\nu[f]=\mathbb{E}_{\nu}\left[f \log \lp \frac{f}{\E_\nu f}\rp\right]$.
% \end{definition}

% We have the following useful lemma for Log-Sobolev Inequality (LSI):
% \begin{lemma}[\cite{BE85}]
% \label{lm:strongly_convexity_LSI}
% If $\nu$ is $\mu$-strongly log-concave, then $\nu$ satisfies LSI with constant $C=1/\mu$.
%If $\nu$ is obtained by restricting a standard Gaussian distribution with variance $\sigma^2$ to some subset $\cK\subset \R^d$, then $\nu$ satisfies the log-Sobolev inequality (\ref{eqn:log-sobolev}) with $C=\sigma^2.$
% \end{lemma}

% The Talagrand transportation inequality is closely related to LSI. 
The proof will make use of one famous inequality: \emph{Talagrand transportation inequality}.
Recall for two probability distributions $\nu_1,\nu_2$, the Wasserstein distance is equivalently defined as $$W_2(\nu_1,\nu_2)=\inf_\Gamma \lp\E_{(x_1,x_2)\sim \Gamma} \norm{x_1-x_2}_2^2\rp^{1/2},$$ where the infimum is over all couplings $\Gamma$ of $\nu_1,\nu_2.$
\begin{theorem}[Talagrand transportation inequality]{\cite[Theorem 1]{OV00}}
\label{thm:TTI}
Let $\d\pi\propto e^{-F(x)}\d x$ be a $\mu$-strongly log-concave probability measure on $\cK\subseteq\R^d$ with finite moments of order 2. 
For all probability measure $\nu$ absolutely continuous w.r.t. $\pi$ and with finite moments of order 2,
we have
\begin{align*}
\mathrm{W}_2(\nu,\pi)\leq\sqrt{\frac{2}{\mu}\mathrm{D}_{KL}(\nu,\pi)}.   
\end{align*}
\end{theorem}

To prove our main result on bounding the generalization error of sampling mechanism, we need the following lemma.

\begin{lemma}[Lemma 7 in \cite{BE02}]
\label{lm:generalization_error_erm}
For any learning algorithm $\cA$ and dataset $\cD=\{s_1,\cdots,s_n\}$ drawn i.i.d from the underlying distribution $\cP$, let $\cD'$ be a neighboring dataset formed by replacing a random element of $\cD$ with a freshly sampled $s'\sim \cP$. If $\cA(\cD)$ is the output of $\cA$ with $\cD$, then
\begin{align*}
    \E_{\cD}[\HF(\cA(\cD))-F(\cA(\cD);\cD)]=\E_{\cD,s'\sim\cP,\cA}\Big[f(\cA(\cD);s')-f(\cA(\cD');s') \Big].
\end{align*}
\end{lemma}

Now we begin to state and prove our main result on the generalization error.
\begin{theorem}
\label{thm:generalization_error}
Suppose $\{f(\cdot,s)\}$ is a family $\mu$-strongly convex functions over $\cK$ such that $f(x;s)-f(x;s')$ is $G$-Lipschitz for all $s,s'$.
For any $k>0$ and dataset $\cD=\{s_1,s_2,\cdots,s_n\}$ drawn i.i.d from the underlying distribution $\cP$, let $\cD'$ be a neighboring dataset formed by replacing a random element of $\cD$ with a freshly sampled $s'\sim \cP$,
\begin{align*}
    \mathrm{W}_{2}(\pi_\cD,\pi_{\cD'})\leq \frac{G}{n\mu}.
\end{align*}
If we sample our solution from density $\pi_\cD(x) \propto e^{-kF(x;\cD)}$, we can bound the excess population loss as:
\begin{align*}
    \E_{\cD,x\sim \pi_\cD}[\HF(x)]-\min_{x\in\cK}\HF(x)\leq \frac{G^2}{\mu n}+ \frac{d}{k}.
\end{align*}
\end{theorem}

\begin{proof}
Recall that
\begin{align*}
    F(x;\cD)=\frac{1}{n}\sum_{s_i\in\cD}f(x;s_i).
\end{align*}
We form a neighboring data set $\cD'$ by replacing a random element of $\cD$ by a freshly sampled $s'\sim \cP.$
Let $\pi_\cD\propto e^{-kF(x;\cD)}$ and $\pi_{\cD'}\propto e^{-kF(x;\cD')}$. By Corollary~\ref{cor:divergence_privacy}, we have
\begin{align*}
    \mathrm{D}_{KL}(\pi_\cD,\pi_{\cD'})\leq \frac{G^2k}{2n^2\mu}.
\end{align*}
By the assumptions, we know both $F(x;\cD)$ and $F(x;\cD')$ are $\mu$-strongly convex and by Theorem~\ref{thm:TTI}, we have
\begin{align*}
    \mathrm{W}_2(\pi_\cD,\pi_{\cD'})\leq \sqrt{\frac{2}{k\mu}\mathrm{D}_{KL}(\pi_\cD,\pi_{\cD'})} \leq \frac{G}{n\mu}.
\end{align*}
By Lemma~\ref{lm:generalization_error_erm} and properties of Wasserstein distance, we have
\begin{align*}
    \E_{\cD,x\sim \pi_{\cD}}[\HF(x)-F(x;\cD)]=&~\E_{\cD,s'\sim \cP}\lb\E_{x\sim \pi_\cD}f(x;s')-\E_{x'\sim \pi_{\cD'}}f(x';s')\rb\\
    =&~\E_{\cD,s'\sim \cP}\lb\E_{x\sim \pi_\cD}\lb f(x;s')-f(x;s'')\rb -\E_{x'\sim \pi_{\cD'}}\lb f(x';s')-f(x';s'')\rb \rb \tag{where $s''$ is chosen arbitrarily, note that $\E_{\cD,x\sim \pi_\cD}[f(x;s'')]=\E_{\cD',x'\sim \pi_{\cD'}}[f(x';s'')]$}\\
    \leq &~G\cdot \mathrm{W}_{2}(\pi_\cD,\pi_{\cD'}) \tag{$f(x;s')-f(x;s'')$ is $G$-Lipschitz}\\
    \leq &~ \frac{G^2}{n\mu}.
\end{align*}
Hence, we know that
\begin{align*}
    \E_{\cD,x\sim\pi_{\cD}}[\hat{F}(x)]-\min_{x\in\cK}\HF(x)\leq &~ \E_{\cD,x\sim\pi_{\cD}}[\hat{F}(x)]-\E_{\cD}[\min_{x\in\cK}F(x;\cD)]\\
    \leq &~ \E_{\cD,x\sim\pi_{\cD}}[\hat{F}(x)-F(x;\cD)]+\E_{\cD,x\sim\pi_{\cD}}[F(x;\cD)-\min_{x\in\cK}F(x;\cD)]\\
    \leq &~ \frac{G^2}{n\mu}+\E_{\cD,x\sim\pi_{\cD}}[F(x;\cD)-\min_{x\in\cK}F(x;\cD)]\\
    \leq &~ \frac{G^2}{n\mu}+\frac{d}{k},
\end{align*}
where the last inequality follows from Lemma~\ref{lm:utility_tech}.
\end{proof}

With the bounds on generalization error, we can get our first result on DP-SCO.
\begin{theorem}[DP-SCO]
\label{thm:dpsco}
Let $\epsilon>0$, $\cK\subseteq \R^d$ be a convex set of diameter $D$ and $\{f(\cdot;s)\}_{s\in\mathcal{D}}$ be a family of convex functions over $\cK$ such that $f(x;s)-f(x;s')$ is $G$-Lipschitz for all $s,s'$.
For any data-set $\cD$ and $k>0$, sampling $x^{(priv)}$ with probability proportional to $\exp\left(-k(F(x;\cD)+\mu\|x\|_2^2/2)\right)$ is $(\epsilon,\delta(\eps))$-differentially private, where
\begin{align*}
 \delta(\eps)\leq \deltacurve{\cN(0,1)}{\cN\lp\frac{G\sqrt{k}}{n\sqrt{\mu}},1\rp}(\epsilon).
\end{align*}
If users in the data-set $\cD$ are drawn i.i.d. from the underlying distribution $\cP$, the excess population loss is bounded by $\frac{G}{n\mu}+\frac{d}{k}+\frac{\mu D^2}{2}$. 
Moreover, if $\{f(\cdot;s)\}_{s\in\mathcal{D}}$ are already $\mu$-strongly convex, then sampling 
$x^{(priv)}$ with probability proportional to $\exp(-kF(x;\cD))$ is $(\eps,\delta(\eps))$-differentially private where 
\begin{align*}
 \delta(\eps)\leq \deltacurve{\cN(0,1)}{\cN\lp\frac{G\sqrt{k}}{n\sqrt{\mu}},1\rp}(\epsilon).
 \end{align*}
 The excess population loss is bounded by $\frac{G}{n\mu}+\frac{d}{k}$.
% \Gopi{Add what happens if $f(\cdot;s)$ are already $\mu$-strongly convex.}
\end{theorem}
% Combining our result on DP-ERM (Theorem~\ref{thm:DPERM_impl}) and the bound on generalization error (Theorem~\ref{thm:generalization_error}), we can get our result about DP-SCO.

\begin{proof}
The first part about privacy is a restatement of our result on DP-ERM (Theorem~\ref{thm:DPERM_impl}).
The excess population loss (See Equation~\eqref{eqn:population_loss}) follows from the bound on generalization error (Theorem~\ref{thm:generalization_error}) and utility guarantee (Lemma~\ref{lm:utility_tech}).
\end{proof}
We give an implementation result of our DP-SCO result.
\begin{theorem}[DP-SCO Implementation]
\label{thm:dpsco_impl}
With same assumptions in Theorem~\ref{thm:dpsco}, and assume $f(\cdot;s)$ is $G$-Lipschitz over $\cK$ for all $s$.
For $0<\delta<\frac{1}{10}$ and $0 < \eps<\frac{1}{10}$, there is an efficient algorithm to solve DP-SCO which has the following guarantees:
\begin{itemize}
    \item The algorithm is $(\eps,\delta)$-differentially private;
    \item The expected population loss is bounded by
    \begin{align*}
        GD\lp\frac{2\sqrt{\log(1/\delta)d}}{\eps n}+\frac{2}{\sqrt{n}}\rp,
    \end{align*}
    where $c>0$ is an arbitrary constant to be chosen. 
    %The expected population loss for strongly-convex case is bounded by
    %\begin{align*}
    %    O(\frac{G^2}{\mu}(\frac{1}{n}+\frac{d\log(1/\delta)}{\eps^2n^2})).
    %\end{align*}
    \item The algorithm takes 
    \begin{align*}
        O\lp\min\lc\frac{\eps^2n^2}{\log(1/\delta)},nd\rc\log^2\lp\frac{\eps nd}{\delta}\rp\rp
    \end{align*}
    queries of the values of $f(\cdot,s_i)$ in expectation and takes the same number of samples from some Gaussian restricted to the convex set $\mathcal{K}$.
\end{itemize}
\end{theorem}
\begin{remark}
As for the non-typical case when $\eps\geq 1/10$, one can use the bound in Theorem~\ref{thm:DPERM_impl} and the bound on generalization error (Theorem~\ref{thm:generalization_error}) .
Particularly, one can achieve expected population loss $O\lp GD\lp\frac{\sqrt{d}/n}{\sqrt{\log(1/\delta)+\eps}-\sqrt{\log(1/\delta)}}+\frac{1}{\sqrt{n}}\rp\rp$.
\end{remark}

\begin{proof}
By Theorem~\ref{thm:dpsco}, sampling from $\exp(-k(F(x;\cD)+\mu \|x\|_2^2/2))$ when $k\leq \frac{\eps^2n^2\mu}{2G^2\log(3/(4\delta))}$ is $(\eps,2\delta/3)$-DP.
Besides, 
we can set  $k=\frac{\mu}{G^2}\min\{\frac{\epsilon^{2}n^{2}}{2\log(3/(4\delta))},2nd\}$ for arbitrarily large constant $c>0$ to make the mechanism $(\eps,2\delta/3)$-differentially private, achieving tight population loss and decrease the running time.
Then the population loss is upper bounded by
\begin{align*}
    \frac{d}{k}+\frac{\mu D^2}{2}+\frac{G^2}{\mu n}=& \frac{G^2}{\mu}\max\lc\frac{2\log(3/(4\delta))d}{\eps^2n^2},\frac{1}{2n}\rc+\frac{\mu D^2}{2}+\frac{G^2}{\mu n}.
\end{align*}
% where the generalization error term $\frac{2G^2}{\mu n}$ follows from bounding Wasserstein distance by Theorem~\ref{thm:generalization_error} and making use of Lemma~\ref{lm:generalization_error_erm} the assumption that $f(;s')$ is $G$-Lipschitz for any $s'$.

By setting $\mu=\frac{G}{D}\sqrt{2(\frac{2\log(3/(4\delta))d}{\eps^2 n^2}+\frac{1}{2n})}$, the population loss is upper bounded by
\begin{align*}
    GD\sqrt{\frac{4\log(3/(4\delta))d}{\eps^2 n^2}+\frac{1}{n}}+GD\sqrt{\frac{1}{n}}\leq GD\lp\frac{2\sqrt{\log(3/(4\delta))d}}{\eps n}+\frac{2}{\sqrt{n}}\rp.
\end{align*}

To make it algorithmic, we also apply Theorem~\ref{thm:sampler} with the accuracy on the total variation distance to be $\min\{\delta/3,\frac{1}{cn^c}\}$ for some large enough constant $c$. This leads to an extra empirical loss and hence we use $\log(1/\delta)$ rather than $\log(3/(4\delta))$ in the final loss term. The runtime follows from Theorem~\ref{thm:sampler}.
\end{proof}

\begin{comment}
\begin{remark}
An alternative way to achieve $O(nd)$ first order oracle (query sub-gradient) complexity is to run proximal point algorithm on non-smooth functions, which can be shown with similar contractility as running stochastic gradient descent on smooth functions. Thus replacing the SGD in \cite{FKT20} by proximal point algorithm, we can achieve $\Tilde{O}(nd)$ gradient queries for DP-SCO. But this method takes gradient queries while we query values, and it requires more memory space $\Omega(d^2)$ compared to our algorithm.
\end{remark}

\end{comment}

\section{Information-theoretic Lower Bound for DP-SCO}
\label{sec:infolower}

 In this section, we prove an information-theoretic lower bound for the query complexity required for DP-SCO (with value queries), which matches (up to some logarithmic terms) the query complexity achieved by our algorithm (in Theorem \ref{thm:dpsco_impl}).
Our proof is similar to the previous works like \cite{ACCD12,DJWW15} with some modifications.

% In fact, the query complexity of our (non-private) sampling scheme and algorithm for solving DP-SCO matches the  information-theoretic minimax lower bounds (up to some logarithmic terms).
% Our proof is similar to the previous works like \cite{ACCD12,DJWW15} with some modifications.
% For completeness, we present the proof here.
%We will demonstrate the result of DP-SCO first.

Before stating the lower bound, we define some notations.
Recall that we are given a set $\cD$ of $n$ samples (users) $\{s_1,\cdots,s_n\}$.
Let $\mathbb{A}_k$ be the collection of all algorithms that observe a sequence of $k$ data points $(Y^1,\cdots,Y^k)$ with $Y^t=f(X^t;S^t)$ where %$s^t$ is an independent sample randomly drawn from the underlying distribution $\mathcal{P}$, 
$S^t\in \cD$ and $X^t\in \cK$ are chosen arbitrarily and adaptively by the algorithm (and possibly using some randomness).

For the lower bound, we only consider linear functions, that is we define $f(x;s)\defeq\langle x,s\rangle$. And let $\cP_G$ be the collection of all distributions such that if $\cP\in\cP_G$, then $\E_{s\sim\cP}\|s\|_2^2\leq G^2$.

% \begin{align*}
%     \cF_G\defeq\{(f,\mathcal{P}):f(x;s)=\langle x,s\rangle,\E_{s\sim \cP}[\|s\|^2_2]\leq G^2\}
% \end{align*}
% be a class of linear functionals.
And we define the optimality gap
\begin{align*}
    \eps_k(\cA,\cP,\cK)\defeq&~ \E_{\cD\sim\cP^n,\cA}[\HF(\hat{x}(\cD))]-\inf_{x\in \cK}\hat{F}(x),
\end{align*}
where $\HF(x) = \E_{s\sim\cP} f(x; s)$, $\hat{x}$ is the output the algorithm $\cA$ given the input dataset $\cD$ and the expectation is over the dataset $\cD\sim \cP^n$ and the randomness of the algorithm $\cA.$ Note that we can rewrite the optimality gap as:
\begin{align*}
    \eps_k(\cA,\cP,\cK)=&~ \E_{\cD\sim\cP^n,\cA}[\HF(\hat{x}(\cD))]-\inf_{x\in \cK}\hat{F}(x)\\
    =& \E_{s\sim\cP}\Big[\E_{\cD\sim\cP^n,\cA}f(\hat{x}(\cD);s)]\Big]-\inf_{x\in\cK}\E_{s\sim\cP}[f(x;s)]\\
    =& \E_{s\sim\cP,\cD\sim\cP^n,\cA}[\hat{x}(\cD)^\top s]-\inf_{x\in\cK}\E_{s\sim \cP}[x^\top s].
\end{align*}
The minimax error is defined by
\begin{align*}
    \eps_k^*(\cP_G,\cK)\defeq \inf_{\cA\in \mathbb{A}_k}\sup_{\mathcal{P}\in \cP_G}\eps_k(\cA,\cP,\cK).
\end{align*}
%where the expectation is taken over the observations $(Y^1,\cdots,Y^k)$ and any additional randomness in $\cA$.

\begin{theorem}
%\Yintat{Make this statement easier to understand. Maybe add something like: In particular, there is no algorithm that use something something and get error something something}
\label{thm:info_bound}
Let $\cK$ be the $\ell_2$ ball of diameter $D$ in $\R^d$, then
\begin{align*}
    \eps_k^*(\cP_G,\cK)\geq \frac{GD}{16}\min\left\{1,\sqrt{\frac{d}{4k}}\right\}.
\end{align*}
In particular, for any (randomized) algorithm $\cA$ which can observe a sequence of data points $(Y^1,\cdots,Y^k)$ with $Y^t=f(X^t;S^t)$ where $S^t\in\cD=\{s_1,s_2,\dots,s_n\}$ and $X^t\in\cK$ are chosen arbitrarily and adaptively by $\cA$,
there exists a distribution $\cP$ over convex functions such that $\E_{s\sim \mathcal{P}}[\|\nabla f(x,s)\|_2^2]\leq G^2$ for all $x\in \cK$, such that the output $\hx$ of the algorithm satisfies
\begin{align*}
    \E_{s\sim\cP}\Big[ \E_{\cD\sim\cP^n,\cA}f(\hx;s)]\Big]-\min_{x\in\cK}\E_{s\sim\cP}[f(x;s)]\geq \frac{GD}{16}\min\left\{1,\sqrt{\frac{d}{4k}}\right\}.
\end{align*}
\end{theorem}

\subsection{Proof of Theorem \ref{thm:info_bound}}
We reduce the optimization problem into a series of binary hypothesis tests.
Recall we are considering linear functions $f(x;s)\defeq\langle x,s\rangle$.
Let $\cV=\{-1,1\}^d$ be a Boolean hyper-cube and for each $v\in \cV$, let $\cN_{v}=\cN(\delta v,\sigma^2I_{d})$ be a Gaussian distribution for some parameters to be chosen such that $\HF_v(x)\defeq \E_{s\sim\cN_v}[f(x;s)]=\delta\langle x,v\rangle$. Note that $$\E_{s\sim \cN_v}[\|\nabla f(x,s)\|_2^2]=\E_{s\sim \cN_v}[\norm{s}_2^2]=(\delta^2+\sigma^2)d.$$ Therefore $G=\sqrt{d(\delta^2+\sigma^2)}.$
%Then by taking $nd$ queries of function like $f(e_i;s_j)$, we can determine the exact values of the $n$ samples $\{s_1,\cdots,s_n\}$.

Clearly the lower bound should scale linearly with $D$. Therefore without loss of generality, we can assume that the diameter $D=2$ and define $\cK=\{x\in\R^d:\|x\|_2\leq 1\}$ to be the unit ball.
As in \cite{ACCD12}, we suppose that $v$ is uniformly sampled from $\cV=\{-1,1\}^d$.
Note that if we can find a good solution to $\HF_v(x)$, we need to determine the signs of vector $v$ well. Particularly, we have the following claim:
\begin{claim}[\cite{DJWW15}]
\label{clm:error_determine_sign}
For each $v\in \cV$, let $x^v$ minimize $\HF_v$ over $\cK$ and obviously we know that $x^v=-v/\sqrt{d}$. 
For any solution $\hat{x}\in \R^d$, we have 
\begin{align*}
    \HF_v(\hat{x})-\HF_v(x^v)\geq\frac{\delta}{2\sqrt{d}}
    \sum_{j=1}^{d}\indicator\{\sign(\hat{x}_j)\neq \sign(x^v_j) 
    \},
\end{align*}
where the function $\sign(\cdot)$ is defined as:
\begin{align*}
    \sign(\hx_j)=\left\{\begin{array}{cc}
       + &  \text{ if } \hx_j>0 \\
         0 & \text{ if } \hx_j=0\\
         - & \text{ otherwise}
         \end{array}
         \right.
\end{align*}
\end{claim}

Claim~\ref{clm:error_determine_sign} provides a method to lower bound the minimax error.
Specifically, 
we define the hamming distance between any two vectors $x,y\in\R^d$ as $d_H(x,y)=\sum_{j=1}\indicator\{\sign(x_j)\neq \sign(y_j)\}$, and
we have
\begin{align}
\label{eq:minimax_to_testing}
    \eps_k^*(\cP_G,\cK)\geq \frac{\delta}{2\sqrt{d}}\{\inf_{\hat{v}}\E[d_H(\hv,v)]\},
\end{align}
where $\hv$ denotes the output of any algorithm mapping from the observation $(Y^1,\cdots,Y^k)$ to $\{-1,1\}^d$, and the probability is taken over the distribution of the underlying $v$, the observation $(Y^1,\cdots,Y^k) $ and any additional randomness in the algorithm.

% Let $S=\{j:v_j=1\}$, and we define the testing error of a solution $\hat{S}$ for estimating $S$ as
% \begin{align*}
%     \E[\hat{S}\Delta S]\defeq\sum_{j=1}^{d}\Pr[\hat{S}_j\neq S_j].
% \end{align*}

By Equation~\eqref{eq:minimax_to_testing}, it suffices to lower bound the value of the testing error $\E[d_H(\hv,v)]$. 
As discussed in \cite{ACCD12,DJWW15}, the randomness in the algorithm can not help, and we can assume the algorithm is deterministic, i.e. $(X^t,S^t)$ is a deterministic function of $Y^{[t-1]}$.\footnote{We use $Y^{[t]}$ to denote the first $t$ observations, i.e. $(Y^1,\cdots,Y^t)$}
The argument is basically based on the easy direction of Yao's principle.
% Indeed, recall $\cV=\{-1,+1\}^d$ is a finite set indexing a subset $\{\HF_v,\cN_v\}$.
% Then
% \begin{align*}
%     \sup_{\mathcal{P}\in \cP_G}\E[\eps_k(\cA,\cP,\cK)]\geq \frac{1}{|\cV|}\sum_{v\in \cV}\E_{\cN_v}[\HF_v(\hx)-\inf_{x\in\cK}\HF_v(x)]].
% \end{align*}
% At each iteration of any algorithm $\cA$, we can write $(x^t,S^t)=H_t(Y^{[t-1]},U^t)$ where we use $U^t$ is a random variable independent of $Y^{[t-1]}$ and denotes the randomness in $\cA$, and $H_t$ is a deterministic function.
% By the properties of expectations, we have
% \begin{align*}
%     &~\frac{1}{|\cV|}\sum_{v\in\cV}\E_{\cN_v}[\HF(\hx)-\inf_{x\in\cK}\HF_v(x)]\\
%     =&~\frac{1}{|\cV|}\sum_{v\in\cV}\E[\E_{\cN_v}[\HF(\hx)-\inf_{x\in\cK}\HF(x)\mid U^{[k]}]]\\
%     \ge & \inf_{u^{[k]}}\frac{1}{|\cV|}\sum_{v\in\cV}\E[\E_{\cN_v}[\HF(\hx)-\inf_{x\in\cK}\HF(x)\mid U^{[k]}=u^{[k]}]],
% \end{align*}
% which implies we can incorporate the best randomness into the algorithm $\cA$ and assume it is deterministic.

Now we continue our proof of the lower bound. We will make use of the property of the Bayes risk.

\begin{comment}
\begin{definition}[Bayes Risk, \cite{LR05}]
Suppose parameter $\theta$ is assumed
known only that it lies in a certain set $\Theta$ and consider a set of probability distributions $P=\{P_\theta:\theta\in\Theta\}$.
Given $\theta$, we can get observations $Y$ and can output an estimator $\hat{\theta}(Y)$ %\Yintat{what is X?} 
by a testing function $\hat{\theta}$.
We use a loss function $l(\theta,\hat{\theta})$ to quantifies the quality of the estimator.
For a prior distribution $\pi$ on $\Theta$, the average risk is defined as
\begin{align*}
    B_{\pi}(\hat{\theta})=\E_{\theta\sim\pi,Y\sim P_\theta}l(\theta,\hat{\theta}).
\end{align*}
The Bayes risk is the minimum that the average risk can achieve, i.e.
\begin{align*}
    B_{\pi}^*=\inf_{\hat{\theta}}B_{\pi}(\hat{\theta}).
\end{align*}
\end{definition}

\begin{lemma}[{\cite[Lemma 1]{ACCD12}}]
\label{lm:Bayes_risk}
Consider the problem of testing hypothesis $H_{-1}:v \sim \mathbb{P}_{-1}$ and $H_1:v\sim \mathbb{P}_1$, where $H_{-1}$ and $H_1$ are occurred with prior probability $\pi_{-1}$ and $\pi_1$ respectively prior to the experiment. Under the 0-1 loss, the Bayes risk $B$ satisfies
\begin{align*}
    B\geq \min(\pi_{-1},\pi_1)(1-\|\mathbb{P}_1-\mathbb{P}_{-1}\|_{\mathrm{TV}}).
\end{align*}
\end{lemma}
\end{comment}

\begin{lemma}[{\cite[Lemma 1]{ACCD12}}]
\label{lm:Bayes_risk}
Consider the problem of testing hypothesis $H_{-1}:v \sim \mathbb{P}_{-1}$ and $H_1:v\sim \mathbb{P}_1$, where $H_{-1}$ and $H_1$ occur with prior probability $\pi_{-1}$ and $\pi_1\defeq1-\pi_{-1}$ respectively prior to the experiment. 
For any algorithm that takes one sample $v$ and outputs $\hat{i}:v\rightarrow\{-1,1\}$, we define the Bayes risk $B$ be the minimum average probability that algorithm fails ($v$ is not sampled from $H_{\hat{i}(v)}$).
That is $B=\inf_{\hat{i}}\pi_{-1}\Pr[\hat{i}(v)=1\mid v\sim \mathbb{P}_{-1}]+\pi_1\Pr[\hat{i}(v)=0\mid v\sim \mathbb{P}_1]$.
Then, we have
\begin{align*}
    B\geq \min(\pi_{-1},\pi_1)(1-\|\mathbb{P}_1-\mathbb{P}_{-1}\|_{\mathrm{TV}}).
\end{align*}
\end{lemma}

\begin{lemma}
\label{lm:error_binary_test}
Suppose that $v$ is uniformly sampled from $\cV=\{-1,1\}^d$, then any estimate $\hat{v}$
%outputted by algorithm which satisfies the Orthogonal Query assumption
obeys
\begin{align*}
    \E[d_H(\hv,v)]\geq
    \frac{d}{2}\lp 1-\frac{\delta\sqrt{k}}{\sigma\sqrt{d}}\rp.
\end{align*}
\end{lemma}

\begin{proof}
Let $\pi_{-1}=\pi_1=1/2$.
For each $j$, define $\mathbb{P}_{-1,j}=\bbP(Y^{[k]}\mid v_j=-1)$ and $\bbP_{1,j}=\bbP(Y^{[k]}\mid v_j=1)$ to be distributions over the observations $(Y^1,\cdots,Y^k)$ conditional on $v_j\neq 1$ and $v_j=1$ respectively.
Let $B_j$ be the Bayes risk of the decision problem for $j$-th coordinate of $v$ between $H_{-1,j}: v_j=-1$ and $H_{1,j}: v_j=1$.
We have that
\begin{align*}
    \E[d_H(\hv,v)]
    \geq& \sum_{j=1}^{d}B_j\\
    \geq& \pi_1\sum_{j=1}^{d}(1-\|\bbP_{1,j}-\bbP_{-1,j}\|_{\mathrm{TV}})\\
    \geq& \frac{d}{2}\lp 1-\frac{1}{\sqrt{d}}\sqrt{\sum_{j=1}^{d}\|\bbP_{1,j}-\bbP_{-1,j}\|^2_{\mathrm{TV}}}\rp,
\end{align*}
where the first inequality follows from the definition of Bayes risk $B_j$, the second inequality follows by Lemma~\ref{lm:Bayes_risk} and the last inequality follows by the Cauchy-Schwartz inequality.

To complete the proof, it suffices to show that
\begin{align}
\label{eq:bounded_TVD}
    \sum_{j=1}^{d}\|\bbP_{1,j}-\bbP_{-1,j}\|_{\mathrm{TV}}^2\leq \frac{\delta^2}{\sigma^2}k.
\end{align}

Assuming Equation~\eqref{eq:bounded_TVD} first, which will be established later.
Then we know that
\begin{align*}
    \E[d_H(\hv,v)]\geq \frac{d}{2}(1-\frac{\delta\sqrt{k}}{\sigma\sqrt{d}}).
\end{align*}
\end{proof}

We will complete the proof of Lemma~\ref{lm:error_binary_test} by showing the following bounded total variation distance.
\begin{claim}
\begin{align*}
    \sum_{j=1}^{d}\|\bbP_{1,j}-\bbP_{-1,j}\|_{\mathrm{TV}}^2\leq \frac{\delta^2}{\sigma^2}k.
\end{align*}
\end{claim}

\begin{proof}
Applying Pinsker's inequality, we know $\|\bbP_{1,j}-\bbP_{-1,j}\|_{\mathrm{TV}}^2\leq \frac{1}{2}\mathrm{D}_{KL}(\bbP_{-1,j}\|\bbP_{1,j})$.
To bound the KL divergence between $\bbP_{-1,j}$ and $\bbP_{1,j}$ over all possible $Y^{[k]}$, consider $v'=(v_1,\cdots,v_{j-1},v_{j+1},\cdots,v_d)$, and define $\bbP_{-1,j,v'}(Y^{[k]})\defeq \bbP(Y^{[k]}\mid v_j=-1,v')$ to be the distribution conditional on $v_j=-1$ and $v'$.
%\gnote{Use $\bbP_{-1,j,v'}$ notation, since it also depends on $j.$}
We have
\begin{align*}
    \bbP_{-1,j}(Y^{[k]})= \sum_{v'}\Pr[v']\bbP_{-1,j,v'}(Y^{[k]}).
\end{align*}

The convexity of the KL divergence suggests that
\begin{align*}
    \mathrm{D}_{KL}(\bbP_{-1,j}\|\bbP_{1,j})\leq \sum_{v'}\Pr[v']\mathrm{D}_{KL}(\bbP_{-1,j,v'}\|\bbP_{1,j,v'}).
\end{align*}
Fixing any possible $v'$, we want to bound the KL divergence $\mathrm{D}_{KL}(\bbP_{-1,j,v'}\|\bbP_{1,j,v'})$.

Recall we are considering deterministic algorithms and
$(X^t,S^t)$ is a deterministic function of $Y^{[t-1]}$.
Let $Q_i\in \R^{d\times k}$ be a (random) matrix, which records the set of points the algorithm queries for the user $s_i$.
Specifically, for $t$-th step, if the algorithm queries $(X^t,S^t)$, then $Q_{i}^{t}=X^t$ if $S^t=s_i$, otherwise $Q_{i}^{t}=0$, where $Q_i^t$ is the $t$-th column of $Q_i$.

As we are considering linear functions, without loss of generality we can assume  $\langle Q_{i}^j,Q_{i}^{j'}\rangle=0$ for each $i$ and any $j\neq j'$, and $\|Q_{i}^t\|_2\in\{0,1\}$ for any $i$ and $t$.
We name this assumption \textsc{Orthogonal Query}.
Roughly speaking, for any algorithm, we can modify it to satisfy the Orthogonal Query.
Whenever the algorithm wants to query some point, we can use Gram–Schmidt process to query another point and satisfy Orthogonal Query, and recover the function value at the original point queried by the algorithm.
% To establish this, observe that for any algorithm $\cA$, we can find another algorithm $\cA'$ which satisfies the Orthogonal Query assumption, and the distributions of the outputs of $\cA$ and $\cA'$ are the same almost everywhere.
% Give an example, suppose $j$ is the minimum integer that $\cA$ has re-queried information of some user, say $s_i$.
% That is $s^{t_1}\neq s^{t_2}$ for $t_1,t_2\leq j-1$ and $S^j=S^t=s_i$ for some $t\leq j-1$.
% Then we can construct $\cA'$ by taking $\cA$ as an black box: for the first $j-1$-th step, $\cA'$ queries $(X^t/\|X^t\|_2,S^t)$ for the first $j-1$ steps and inputs the values of $f(X^t/\|X^t\|_2,S^t)\cdot \|X^t\|_2$ to $\cA$.
% For $j$th step, $\cA'$ queries $(z';S^t)$ where $z'=\frac{x^j-\frac{\langle x^j,x^i\rangle }{\|x^i\|^2_2}\cdot x^i}{\|z-\frac{\langle x^j,x^i\rangle }{\|x^i\|^2_2}\cdot x^i\|_2}$, then recovers the value of $f(x^j;s^j)$ and inputs it to $\cA$ to observe what's the next query.
% $\cA'$ can repeat this procedure, and output whatever $\cA$ outputs finally.

By the chain-rule of KL-divergence, if we define $P_{-1,j,v'}(Y^t\mid Y^{[t-1]})$ to be the distribution of $t$th observation $Y^t$ conditional on $v'$, $v_j=-1$ and $Y^{[t-1]}$, then we have
\begin{align*}
    \mathrm{D}_{KL}(\bbP_{-1,j,v'}\|\bbP_{1,j,v'})=\sum_{t=1}^{k}\int_{\cY^{t-1}} \mathrm{D}_{KL}(P_{-1,j,v'}(Y^t\mid Y^{[t-1]}=y)\|P_{1,j,v'}(Y^t\mid Y^{[t-1]}=y)\d P_{-1,j,v'}(y).
\end{align*}

Fix $Y^{[t-1]}$ such that $Y^{[t-1]}=y$.
Since the algorithm is deterministic and $(X^t,S^t)$ is fixed given $Y^{[t-1]}$.
Let $S^t=s_i$ so $X^t=Q_i^t$.
%Denote the choice of the algorithm is $(Q_{i}^t,s_i)$ for $t$-th step conditional on $Y^{[t-1]}=y$.

Note that the $n$ users in $\cD$ are i.i.d. sampled.
Then $\mathrm{D}_{KL}(P_{-1,j,v'}(Y^t\mid Y^{[t-1]}=y)\|P_{1,j,v'}(Y^t\mid Y^{[t-1]}=y)$ only depends on the randomness of $s_i$ and the first $t$ columns of $Q_{i}$, which is denoted by $Q_{i}^{[t]}$.
We use $Y^{t}_{j}$ to denote the observation corresponding to user $s_j$ for the $t$th query (if $S^t\neq s_j$, we have $Y^{t}_j=0$).
Note that the observation $Y^{[t]}_i=Q_i^{[t]\top} s_i$ where $s_i\sim \cN(\delta v,\sigma^2 I_d)$. Then we know $Y^{[t]}_i$ is normally distributed with mean $\delta Q_{i}^{[t]\top} v$ and co-variance $\sigma^2 Q_{i}^{[t]\top} Q_{i}^{[t]}$.

Recall that the KL divergence between two normal distributions is $\mathrm{D}_{KL}(\cN(\mu_1,\Sigma)\|\cN(\mu_2,\Sigma))=\frac{1}{2}(\mu_1-\mu_2)^{\top}\Sigma^{-1}(\mu_1-\mu_2)$. Recall that we have the Orthogonal Query assumption and thus $Q_{i}^{[t]\top} Q_{i}^{[t]}\in\{0,1\}^{t\times t}$ is a diagonal matrix.
By the conditional distributions of Gaussian, we know $Y^t_i$ only depends on the $Q_{i}^t$ and it is independent of $Q_{i}^{[t-1]}$.

Hence we have
\begin{align*}
    &\mathrm{D}_{KL}(P_{-1,j,v'}(Y^t\mid Y^{[t-1]}=y)\|P_{1,j,v'}(Y^t\mid Y^{[t-1]}=y))\\
    =&\mathrm{D}_{KL}(P_{-1,j,v'}(Y^t_i\mid Y^{[t-1]}=y)\|P_{1,j,v'}(Y^t_i\mid Y^{[t-1]}=y))\\
    =&\frac{1}{2} (2\delta Q_{i}^t(j))^2/\sigma^2,
\end{align*}
where $Q_{i}^t(j)$ is the $j$-th coordinate of $Q_{i}^t$.
Summing over the terms, one has
\begin{align*}
\sum_{j=1}^{d}\|\bbP_{1,j}-\bbP_{-1,j}\|_{\mathrm{TV}}^2\leq &
    \frac{1}{2}\mathrm{D}_{KL}(\bbP_{-1,j}\|\bbP_{1,j})\\
    \leq&\frac{1}{2} \sum_{t=1}^{k}\sum_{j=1}^{d}\sum_{i=1}^{n}\E[\frac{1}{2} (2\delta Q_{i}^t(j))^2/\sigma^2] \\
    \leq& \frac{\delta^2}{\sigma^2}k,
\end{align*}
where the last line follows from the fact that for each $t$,$ \sum_{i=1}^{n}\|Q_{i}^t\|_2^2=\sum_{i=1}^{n}\sum_{j=1}^{d}(Q_{i}^t(j))^2=1$ as we only query one user for $t$-th step.

This completes the proof.
\end{proof}

Having Lemma~\ref{lm:error_binary_test}, we can complete the proof of Theorem~\ref{thm:info_bound}.

\begin{proof}{of Theorem~\ref{thm:info_bound}.}
As discussed before, we know
\begin{align*}
    \HF_v(\hat{x})-\HF_v(x^v)\geq\frac{\delta}{2\sqrt{d}}\sum_{j=1}^{d}\indicator\{\sign(\hat{x}_j)\neq \sign(x^v_j) \},
\end{align*}
and hence we know that
\begin{align*}
     \eps_k^*(\cP_G,\cK)\geq & \frac{\delta}{2\sqrt{d}}\inf_{\hat{v}}\E[d_H(\hat{v},v)]\\
    \geq & \frac{\delta \sqrt{d}}{4}\lp 1-\frac{\delta\sqrt{k}}{\sigma\sqrt{d}}\rp,
\end{align*}
where the last line follows from Lemma~\ref{lm:error_binary_test}.
%As we assume $\cK$ is the unit ball and thus $D=2$.
We now set $\delta=\frac{\sigma \sqrt{d}}{2\sqrt{k}}$ and $\sigma=\frac{G}{\sqrt{d+d^2/4k}}$, so that $d(\sigma^2+\delta^2)=G^2$.
Hence one has
\begin{align*}
 \eps_k^*(\cP_G,\cK)
\geq \frac{\delta\sqrt{d}}{8} 
= \frac{D\delta\sqrt{d}}{16}
= \frac{GD}{16\sqrt{1+\frac{4k}{d}}}
\ge \frac{GD}{16}\min\left\{1,\sqrt{\frac{d}{4k}}\right\}.
%\ge GD\frac{\min\{\sqrt{d},\sqrt{k}\}}{40\sqrt{k}}.
\end{align*}
%\gnote{Check the last line or modify the lower bound to $ \frac{GD}{16}\min\{1,\sqrt{\frac{d}{4k}}\}$ everywhere.}
Thus we complete the proof.
\end{proof}

\begin{corollary}[Lower bound for DP-SCO]
\label{cor:DPSCOlower}
For any (non-private) algorithm which makes less than $O\lp\min\{\frac{\eps^2n^2}{\log(1/\delta)},nd\}\rp$ function value queries, there exist a convex domain $\cK\subset \R^d$ of diameter $D$, a distribution $\cP$ supported on $G$-Lipschitz linear functions $f(x;s)\defeq\langle x,s\rangle$, such that the output $\hx$ of the algorithm satisfies that
\begin{align*}
    \E_{s\sim\cP}[\langle \hx,s\rangle]-\min_{x\in\cK}\E_{s\sim\cP}[\langle x,s\rangle]\geq \Omega\lp \frac{G D}{\sqrt{1+\log(n)/d}} \cdot \min\lc\frac{\sqrt{\log(1/\delta)d}}{\eps n}+\frac{1}{\sqrt{n}},1\rc \rp.
\end{align*}
\end{corollary}
\begin{proof}
% WLOG, we can assume that $\frac{\sqrt{\log(1/\delta)d}}{\eps n}+\frac{1}{\sqrt{n}}=O(1)$. This is a reasonable assumption as if $\frac{\sqrt{\log(1/\delta)d}}{\eps n}+\frac{1}{\sqrt{n}}=\Omega(1)$, uniformly randomly output a point in $\cK$ is a good DP solution.

Note that Theorem~\ref{thm:info_bound} almost gives us what we want, except that the Lipschitz constant of the functions in the hard distribution is bounded only on average by $G$. To get distributions over $G$-Lipschitz functions, we just condition on the bad event not happening.

Recall that we are considering the set of distributions $\cN_v=\cN(\delta v,\sigma^2 I_d)$ for which $\E_{s\sim\cN_v}\|s\|_2^2\le G^2=d(\delta^2+\sigma^2)$.
And we proved that 
$\inf_{\cA\in \mathbb{A}_k}\sup_{v\in \cV}\E_{s\sim \cN_v,\cA}[\HF_v(\hat{x}_k)-\HF_v^*]\ge \frac{GD}{16}\min\left\{1,\sqrt{\frac{d}{4k}}\right\}$ in Theorem~\ref{thm:info_bound}, where $\hat{x}_k$ is the output of $\cA$ with $k$ observations $Y^{[k]}$.
%\gnote{What is $f,\cF_G,\hat{x}_k$ in the above equation?}
To prove Corollary~\ref{cor:DPSCOlower}, we need to modify the distribution of $s$ to satisfy the Lipschitz continuity.

In particularly, for some constant $c$, we know 
\begin{align*}
    &\E[\HF_v(\hat{x}_k)-\HF_v^*]\\
    =&\E\Big[ \HF_v(\hat{x}_k)-\HF_v^*\mid \max_{s_i\in \cD}\|s_i\|_2\leq cG\sqrt{1+\log(nd)/d}\Big]\Pr\Big[\max_{s_i\in \cD}\|s_i\|_2\leq cG \sqrt{1+\log(nd)/d}\Big]+\\
    &~~\E\Big[ \HF_v(\hat{x}_k)-\HF_v^*\mid \max_{s_i\in \cD}\|s_i\|_2> cG\sqrt{1+\log(nd)/d}\Big]\Pr\Big[\max_{s_i\in \cD}\|s_i\|_2> cG\sqrt{1+\log(nd)/d}\Big].
\end{align*}
By the concentration of spherical Gaussians, we know if $s\sim\cN(\delta v,\sigma^2 I_d)$, then 
\begin{align*}
    \Pr\Big[\|s-\delta v\|_2^2\leq \sigma^2 d(1+2\sqrt{\ln(1/\eta)/d}+2\ln(1/\eta)/d)\Big]\geq 1-\eta.
\end{align*}
% \Gopi{Do we really need to lose $\log(nd)$ factor here? looks like a constant is enough...}
% \Daogao{We want $\eta\leq 1/\poly(nd)$. I use log in case that $n\gg d$. maybe I just modify it.}
We can choose the constant $c$ large enough, such that $\Pr[\max_{s_i\in \cD}\|s_i\|_2\leq cG\sqrt{1+\log(nd)/d}]\geq 1-1/\poly(nd)$, which implies
\begin{align*}
    \inf_{\cA\in \mathbb{A}_k}\sup_{v\in \cV}\E_{\cD \sim \cN_v^n,\cA}\Big[\HF_v(\hat{x}_k)-\HF_v^*\mid \max_{s_i\in \cD}\|s_i\|_2\leq cG\sqrt{1+\log(nd)/d}\Big]\geq \Omega(GD\frac{\min\{\sqrt{d},\sqrt{k}\}}{\sqrt{k}}).
\end{align*}
If we use the distributions conditioned on $\max_{s_i\in \cD}\|s_i\|_2\leq cG\sqrt{1+\log(nd)/d}$ rather than the Gaussians, and scale the constant to satisfy the assumption on Lipschitz continuity, we can prove the statement.
Particularly, let $G'=cG(\sqrt{1+\log(nd)/d})$. 
If the algorithm can only make $k=O\lp\min\{\frac{\eps^2n^2}{\log(1/\delta)},nd\}\rp$ observations, we know 
\begin{align*}
    &\inf_{\cA\in \mathbb{A}_k}\sup_{v\in \cV}\E_{\cD \sim \cN_v^n,\cA}\Big[\HF_v(\hat{x}_k)-\HF_v^*\mid \max_{s_i\in \cD}\|s_i\|_2\leq G'\Big]\\
    \geq&
    \Omega\lp GD\cdot \min\lc(\frac{\sqrt{\log(1/\delta)d}}{\eps n}+\frac{1}{\sqrt{n}}),1\rc \rp  \\
    =&\Omega\lp \frac{G' D}{\sqrt{1+\log(nd)/d}}\cdot \min\lc\frac{\sqrt{\log(1/\delta)d}}{\eps n}+\frac{1}{\sqrt{n}},1\rc  \rp,
\end{align*}
% \Gopi{There is a gap in the proof. We don't get the required result exactly.}
which proves the lower bound claimed in the Corollary statement.
\end{proof}

\begin{corollary}[Lower bound for sampling scheme]
\label{cor:Samplinglower}
Given any $G > 0$ and $\mu > 0$. For any algorithm which takes function values queries less than $O\lp\frac{G^2}{\mu}/(1+\log(G^2/\mu)/d)\rp$ times, there is a family of $G$-Lipschitz linear functions $\{f_i(x)\}_{i\in I}$ defined on some $\ell_2$ ball $\cK\subset\R^d$, such that the total variation distance between the distribution of the output of the algorithm and the distribution proportional to $\exp(-\E_{i\in I}f_i(x)-\mu \|x\|^2 / 2)$ is at least $\min(1/2, \sqrt{d \mu / G^2})$.
\end{corollary}
\begin{proof}
By a similar argument in the proof of Corollary~\ref{cor:DPSCOlower}, for any algorithm which can only make $k$ observations, there are a family of $G$-Lipschitz linear functions restricted on an $\ell_2$ ball $\cK$ of diameter $D$ centered at $\mathbf{0}$ such that  
\begin{align}
\label{eq:lower_sampling_SCO}
    \E\Big[\HF_v(\hat{x}_k)-\HF_v^*\Big]
    \ge&\Omega\lp \frac{G D}{\sqrt{1+\log(k)/d}} \cdot \min\lc \sqrt{\frac{d}{k}},1\rc \rp,
\end{align}
where $\HF_v^*=\min_{x\in\cK}\HF_v(x)$ and $\hx_k\in\cK$ is the output of $\cA$.

%Without loss of generality, we can shift the functions and assume the zero point $\mathbf{0}$ is the center of $\cK$.
Suppose we have a sampling algorithm that takes $k$ queries. We use it to sample from $x^{(sol)}$ proportional to $p(x):=\exp(-\HF_v(x)-\frac{\mu}{2} \|x\|^2)$ on $\cK$ with total variation distance $\eta\leq \min(1/2, \sqrt{d \mu / G^2})$. 
% Note that $\HF_v(x)$ is some $G$-Lipschitz linear function, and by Lemma~\ref{lm:utility_tech}, we know $\E[\|x^{(sol)}\|]\leq O(\sqrt{d/\mu})$.
% To ensure $x^{(sol)}$ is bounded, we can define the bounded variant $\overline{x^{(sol)}} = x^{(sol)}$ if $\|x^{(sol)}\|^2 \leq O(d \log(1/\eta) / \mu)$ and $\mathbf{0}$ otherwise. Note that $p$ is $\mu$-strongly convex, we have $\|x\|^2 \leq O(d \log(1/\eta)/ \mu)$ with probability $1-\eta$ by Lemma~\ref{lem:gaussian_concentration} and hence $\overline{x^{(sol)}}$ has total variation distance $2 \eta$ from $p$.

Lemma \ref{lm:utility_tech} shows that
\begin{align*}
    \E[\HF_{v}(x^{(sol)})+\frac{\mu}{2}\|x^{(sol)}\|^{2}]\leq\min_{x\in \cK}\lp\HF_{v}(x)+\frac{\mu}{2}\|x\|^{2}\rp+O(d) + O(\eta) \cdot (GD+\mu D^2),
\end{align*}
where the last term involving $\eta$ is due to the total variation distance between $x^{(sol)}$ and $p$. Setting $D=\sqrt{d/\mu}$ and using the diameter of $\cK$ is $D$ and $\eta \leq \min(1/2, \sqrt{d \mu / G^2})$, we have
\begin{align*}
\E[\HF_{v}(x^{(sol)})] & \leq\min_{x\in\cK}\HF_{v}(x)+\frac{\mu}{2}D^{2}+O(d+\eta\cdot (GD+\mu D^2))\\
& \leq \min_{x\in\cK}\HF_{v}(x)+O(d).
\end{align*}
Note that we set $D = \sqrt{d/\mu}$. Comparing with \eqref{eq:lower_sampling_SCO}, we have
\[
\frac{G\sqrt{d/\mu}}{\sqrt{1+\log(k)/d}}\min\left\{ \sqrt{\frac{d}{k}},1\right\} \leq O(d).
\]

If $d\leq G^{2}/\mu\leq \exp(d)$, we have
\[
G\sqrt{d/\mu}\sqrt{\frac{d}{k}}\leq O(d)
\]
and hence $k=\Omega(G^{2}/\mu)$.
If $G^{2}/\mu\geq\exp(d)$,
we have
\[
\frac{G\sqrt{d/\mu}}{\sqrt{\log(k)/d}}\sqrt{\frac{d}{k}}\leq O(d)
\]
and hence $k=\Omega(\frac{G^{2}d/\mu}{\log(G^{2}/\mu)})$.
If $G^{2}/\mu\leq d$, we can construct our function only on the first
$O(G^{2}/\mu)$ dimensions to get a lower bound $k=\Omega(G^{2}/\mu).$
Combining all cases gives the result.
\end{proof}

\addcontentsline{toc}{section}{References}
\bibliographystyle{alpha}
\bibliography{ref}

\newcommand{\etalchar}[1]{$^{#1}$}
\begin{thebibliography}{MMW{\etalchar{+}}21}

\bibitem[Abo16]{Abo16}
John~M. Abowd.
\newblock The challenge of scientific reproducibility and privacy protection
  for statistical agencies.
\newblock {\em Technical report, Census Scientific Advisory Committee}, 2016.

\bibitem[ACCD12]{ACCD12}
Ery Arias-Castro, Emmanuel~J Candes, and Mark~A Davenport.
\newblock On the fundamental limits of adaptive sensing.
\newblock {\em IEEE Transactions on Information Theory}, 59(1):472--481, 2012.

\bibitem[AFKT21]{AFKT21}
Hilal Asi, Vitaly Feldman, Tomer Koren, and Kunal Talwar.
\newblock Private stochastic convex optimization: Optimal rates in l1 geometry.
\newblock In {\em International Conference on Machine Learning}, pages
  393--403. PMLR, 2021.

\bibitem[AKRS19]{AKR+19}
Jordan Awan, Ana Kenney, Matthew Reimherr, and Aleksandra Slavkovi{\'c}.
\newblock Benefits and pitfalls of the exponential mechanism with applications
  to hilbert spaces and functional pca.
\newblock In {\em International Conference on Machine Learning}, pages
  374--384. PMLR, 2019.

\bibitem[App17]{Apple17}
Differential Privacy~Team Apple.
\newblock Learning with privacy at scale.
\newblock {\em Technical report, Apple}, 2017.

\bibitem[BE02]{BE02}
Olivier Bousquet and Andr{\'e} Elisseeff.
\newblock Stability and generalization.
\newblock {\em The Journal of Machine Learning Research}, 2:499--526, 2002.

\bibitem[BEL18]{BEL18}
S{\'e}bastien Bubeck, Ronen Eldan, and Joseph Lehec.
\newblock Sampling from a log-concave distribution with projected langevin
  monte carlo.
\newblock {\em Discrete \& Computational Geometry}, 59(4):757--783, 2018.

\bibitem[BEM{\etalchar{+}}17]{BEM+17}
Andrea Bittau, {\'U}lfar Erlingsson, Petros Maniatis, Ilya Mironov, Ananth
  Raghunathan, David Lie, Mitch Rudominer, Ushasree Kode, Julien Tinnes, and
  Bernhard Seefeld.
\newblock Prochlo: Strong privacy for analytics in the crowd.
\newblock In {\em Proceedings of the 26th Symposium on Operating Systems
  Principles}, pages 441--459, 2017.

\bibitem[BFGT20]{bfgt20}
Raef Bassily, Vitaly Feldman, Crist{\'o}bal Guzm{\'a}n, and Kunal Talwar.
\newblock Stability of stochastic gradient descent on nonsmooth convex losses.
\newblock {\em Advances in Neural Information Processing Systems}, 33, 2020.

\bibitem[BFTT19]{bftt19}
Raef Bassily, Vitaly Feldman, Kunal Talwar, and Abhradeep~Guha Thakurta.
\newblock Private stochastic convex optimization with optimal rates.
\newblock In {\em Advances in Neural Information Processing Systems}, pages
  11282--11291, 2019.

\bibitem[BGN21]{bgn21}
Raef Bassily, Crist{\'o}bal Guzm{\'a}n, and Anupama Nandi.
\newblock Non-euclidean differentially private stochastic convex optimization.
\newblock In {\em Conference on Learning Theory}, pages 474--499. PMLR, 2021.

\bibitem[BNS13]{BNS13}
Amos Beimel, Kobbi Nissim, and Uri Stemmer.
\newblock Private learning and sanitization: Pure vs. approximate differential
  privacy.
\newblock In {\em Approximation, Randomization, and Combinatorial Optimization.
  Algorithms and Techniques}, pages 363--378. Springer, 2013.

\bibitem[BST14]{BST14}
Raef Bassily, Adam Smith, and Abhradeep Thakurta.
\newblock Private empirical risk minimization: Efficient algorithms and tight
  error bounds.
\newblock In {\em 2014 IEEE 55th Annual Symposium on Foundations of Computer
  Science}, pages 464--473. IEEE, 2014.

\bibitem[BV19]{BV19}
Victor Balcer and Salil Vadhan.
\newblock Differential privacy on finite computers.
\newblock {\em Journal of Privacy and Confidentiality}, 9:2, 2019.

\bibitem[BW18]{BalleW18}
Borja Balle and Yu-Xiang Wang.
\newblock Improving the gaussian mechanism for differential privacy: Analytical
  calibration and optimal denoising.
\newblock In {\em International Conference on Machine Learning}, pages
  403--412, 2018.

\bibitem[CDJB20]{CDJB20}
Niladri Chatterji, Jelena Diakonikolas, Michael~I Jordan, and Peter Bartlett.
\newblock Langevin monte carlo without smoothness.
\newblock In {\em International Conference on Artificial Intelligence and
  Statistics}, pages 1716--1726. PMLR, 2020.

\bibitem[CDWY20]{CDW+20}
Yuansi Chen, Raaz Dwivedi, Martin~J Wainwright, and Bin Yu.
\newblock Fast mixing of metropolized hamiltonian monte carlo: Benefits of
  multi-step gradients.
\newblock {\em J. Mach. Learn. Res.}, 21:92--1, 2020.

\bibitem[Che21]{C21}
Yuansi Chen.
\newblock An almost constant lower bound of the isoperimetric coefficient in
  the kls conjecture.
\newblock {\em Geometric and Functional Analysis}, 31(1):34--61, 2021.

\bibitem[CKS20]{CKS20}
Cl{\'e}ment~L Canonne, Gautam Kamath, and Thomas Steinke.
\newblock The discrete gaussian for differential privacy.
\newblock {\em Advances in Neural Information Processing Systems},
  33:15676--15688, 2020.

\bibitem[CM08]{CM08}
Kamalika Chaudhuri and Claire Monteleoni.
\newblock Privacy-preserving logistic regression.
\newblock In {\em NIPS}, volume~8, pages 289--296. Citeseer, 2008.

\bibitem[CMS11]{cms11}
Kamalika Chaudhuri, Claire Monteleoni, and Anand~D Sarwate.
\newblock Differentially private empirical risk minimization.
\newblock {\em Journal of Machine Learning Research}, 12(3), 2011.

\bibitem[CSS13]{CSS13}
Kamalika Chaudhuri, Anand~D Sarwate, and Kaushik Sinha.
\newblock A near-optimal algorithm for differentially-private principal
  components.
\newblock {\em Journal of Machine Learning Research}, 14, 2013.

\bibitem[CV19]{CV19}
Zongchen Chen and Santosh~S Vempala.
\newblock Optimal convergence rate of hamiltonian monte carlo for strongly
  logconcave distributions.
\newblock In {\em Approximation, Randomization, and Combinatorial Optimization.
  Algorithms and Techniques (APPROX/RANDOM 2019)}. Schloss
  Dagstuhl-Leibniz-Zentrum fuer Informatik, 2019.

\bibitem[Dal17]{D17}
Arnak~S Dalalyan.
\newblock Theoretical guarantees for approximate sampling from smooth and
  log-concave densities.
\newblock {\em Journal of the Royal Statistical Society: Series B (Statistical
  Methodology)}, 79(3):651--676, 2017.

\bibitem[DJWW15]{DJWW15}
John~C Duchi, Michael~I Jordan, Martin~J Wainwright, and Andre Wibisono.
\newblock Optimal rates for zero-order convex optimization: The power of two
  function evaluations.
\newblock {\em IEEE Transactions on Information Theory}, 61(5):2788--2806,
  2015.

\bibitem[DKL18]{DKL18}
Etienne De~Klerk and Monique Laurent.
\newblock Comparison of lasserre’s measure-based bounds for polynomial
  optimization to bounds obtained by simulated annealing.
\newblock {\em Mathematics of Operations Research}, 43(4):1317--1325, 2018.

\bibitem[DKM{\etalchar{+}}06]{DKMMN06}
Cynthia Dwork, Krishnaram Kenthapadi, Frank McSherry, Ilya Mironov, and Moni
  Naor.
\newblock Our data, ourselves: Privacy via distributed noise generation.
\newblock In {\em Annual International Conference on the Theory and
  Applications of Cryptographic Techniques}, pages 486--503. Springer, 2006.

\bibitem[DKY17]{DKY17}
Bolin Ding, Janardhan Kulkarni, and Sergey Yekhanin.
\newblock Collecting telemetry data privately.
\newblock {\em Advances in Neural Information Processing Systems}, 30, 2017.

\bibitem[DMM19]{DMM19}
Alain Durmus, Szymon Majewski, and B{\l}a{\.z}ej Miasojedow.
\newblock Analysis of langevin monte carlo via convex optimization.
\newblock {\em The Journal of Machine Learning Research}, 20(1):2666--2711,
  2019.

\bibitem[DMNS06]{DMNS06}
Cynthia Dwork, Frank McSherry, Kobbi Nissim, and Adam Smith.
\newblock Calibrating noise to sensitivity in private data analysis.
\newblock In {\em Theory of cryptography conference}, pages 265--284. Springer,
  2006.

\bibitem[DRS19]{dong2019gaussian}
Jinshuo Dong, Aaron Roth, and Weijie~J Su.
\newblock Gaussian differential privacy.
\newblock {\em Journal of the Royal Statistical Society: Series B (Statistical
  Methodology)}, 2019.

\bibitem[EPK14]{EPK14}
{\'U}lfar Erlingsson, Vasyl Pihur, and Aleksandra Korolova.
\newblock Rappor: Randomized aggregatable privacy-preserving ordinal response.
\newblock In {\em Proceedings of the 2014 ACM SIGSAC conference on computer and
  communications security}, pages 1054--1067, 2014.

\bibitem[Fel16]{Fel16}
Vitaly Feldman.
\newblock Generalization of erm in stochastic convex optimization: The
  dimension strikes back.
\newblock {\em Advances in Neural Information Processing Systems},
  29:3576--3584, 2016.

\bibitem[FKT20]{FKT20}
Vitaly Feldman, Tomer Koren, and Kunal Talwar.
\newblock Private stochastic convex optimization: optimal rates in linear time.
\newblock In {\em Proceedings of the 52nd Annual ACM SIGACT Symposium on Theory
  of Computing}, pages 439--449, 2020.

\bibitem[FTS17]{fts17}
Kazuto Fukuchi, Quang~Khai Tran, and Jun Sakuma.
\newblock Differentially private empirical risk minimization with input
  perturbation.
\newblock In {\em International Conference on Discovery Science}, pages 82--90.
  Springer, 2017.

\bibitem[GT20]{GT20}
Arun Ganesh and Kunal Talwar.
\newblock Faster differentially private samplers via r{\'e}nyi divergence
  analysis of discretized langevin mcmc.
\newblock {\em Advances in Neural Information Processing Systems},
  33:7222--7233, 2020.

\bibitem[GTU22]{GTU22}
Arun Ganesh, Abhradeep Thakurta, and Jalaj Upadhyay.
\newblock Langevin diffusion: An almost universal algorithm for private
  euclidean (convex) optimization.
\newblock {\em arXiv preprint arXiv:2204.01585}, 2022.

\bibitem[HK12]{HK12}
Zhiyi Huang and Sampath Kannan.
\newblock The exponential mechanism for social welfare: Private, truthful, and
  nearly optimal.
\newblock In {\em 2012 IEEE 53rd Annual Symposium on Foundations of Computer
  Science}, pages 140--149. IEEE, 2012.

\bibitem[HRS16]{HRS16}
Moritz Hardt, Ben Recht, and Yoram Singer.
\newblock Train faster, generalize better: Stability of stochastic gradient
  descent.
\newblock In {\em International Conference on Machine Learning}, pages
  1225--1234. PMLR, 2016.

\bibitem[HT10]{HT10}
Moritz Hardt and Kunal Talwar.
\newblock On the geometry of differential privacy.
\newblock In {\em Proceedings of the forty-second ACM symposium on Theory of
  computing}, pages 705--714, 2010.

\bibitem[INS{\etalchar{+}}19]{ins+19}
Roger Iyengar, Joseph~P Near, Dawn Song, Om~Thakkar, Abhradeep Thakurta, and
  Lun Wang.
\newblock Towards practical differentially private convex optimization.
\newblock In {\em 2019 IEEE Symposium on Security and Privacy (SP)}, pages
  299--316. IEEE, 2019.

\bibitem[JLLV21]{JLLV21}
He~Jia, Aditi Laddha, Yin~Tat Lee, and Santosh Vempala.
\newblock Reducing isotropy and volume to kls: an $o(n^3 \psi^2)$ volume
  algorithm.
\newblock In {\em Proceedings of the 53rd Annual ACM SIGACT Symposium on Theory
  of Computing}, pages 961--974, 2021.

\bibitem[JT14]{jt14}
Prateek Jain and Abhradeep~Guha Thakurta.
\newblock (near) dimension independent risk bounds for differentially private
  learning.
\newblock In {\em International Conference on Machine Learning}, pages
  476--484. PMLR, 2014.

\bibitem[KCK{\etalchar{+}}18]{KCK+18}
Yu-Hsuan Kuo, Cho-Chun Chiu, Daniel Kifer, Michael Hay, and Ashwin
  Machanavajjhala.
\newblock Differentially private hierarchical count-of-counts histograms.
\newblock {\em Proceedings of the VLDB Endowment}, 11(11), 2018.

\bibitem[KD99]{KD99}
Jayesh~H Kotecha and Petar~M Djuric.
\newblock Gibbs sampling approach for generation of truncated multivariate
  gaussian random variables.
\newblock In {\em 1999 IEEE International Conference on Acoustics, Speech, and
  Signal Processing. Proceedings. ICASSP99 (Cat. No. 99CH36258)}, volume~3,
  pages 1757--1760. IEEE, 1999.

\bibitem[KJ16]{kj16}
Shiva~Prasad Kasiviswanathan and Hongxia Jin.
\newblock Efficient private empirical risk minimization for high-dimensional
  learning.
\newblock In {\em International Conference on Machine Learning}, pages
  488--497. PMLR, 2016.

\bibitem[KLL21]{KLL21}
Janardhan Kulkarni, Yin~Tat Lee, and Daogao Liu.
\newblock Private non-smooth erm and sco in subquadratic steps.
\newblock {\em Advances in Neural Information Processing Systems}, 34, 2021.

\bibitem[KT13]{KT13}
Michael Kapralov and Kunal Talwar.
\newblock On differentially private low rank approximation.
\newblock In {\em Proceedings of the twenty-fourth annual ACM-SIAM symposium on
  Discrete algorithms}, pages 1395--1414. SIAM, 2013.

\bibitem[KV06]{KV06}
Adam~Tauman Kalai and Santosh Vempala.
\newblock Simulated annealing for convex optimization.
\newblock {\em Mathematics of Operations Research}, 31(2):253--266, 2006.

\bibitem[LC21]{LC21}
Jiaming Liang and Yongxin Chen.
\newblock A proximal algorithm for sampling from non-smooth potentials.
\newblock {\em arXiv preprint arXiv:2110.04597}, 2021.

\bibitem[Led99]{Led99}
Michel Ledoux.
\newblock Concentration of measure and logarithmic sobolev inequalities.
\newblock In {\em Seminaire de probabilites XXXIII}, pages 120--216. Springer,
  1999.

\bibitem[LL21]{LL21}
Daogao Liu and Zhou Lu.
\newblock Curse of dimensionality in unconstrained private convex erm.
\newblock {\em arXiv preprint arXiv:2105.13637}, 2021.

\bibitem[LST20]{LST20}
Yin~Tat Lee, Ruoqi Shen, and Kevin Tian.
\newblock Logsmooth gradient concentration and tighter runtimes for
  metropolized hamiltonian monte carlo.
\newblock In {\em Conference on Learning Theory}, pages 2565--2597. PMLR, 2020.

\bibitem[LST21]{LST21}
Yin~Tat Lee, Ruoqi Shen, and Kevin Tian.
\newblock Structured logconcave sampling with a restricted gaussian oracle.
\newblock In {\em Conference on Learning Theory}, pages 2993--3050. PMLR, 2021.

\bibitem[LSV18]{LSV18}
Yin~Tat Lee, Zhao Song, and Santosh~S Vempala.
\newblock Algorithmic theory of odes and sampling from well-conditioned
  logconcave densities.
\newblock {\em arXiv preprint arXiv:1812.06243}, 2018.

\bibitem[LT19]{LT19}
Jingcheng Liu and Kunal Talwar.
\newblock Private selection from private candidates.
\newblock In {\em Proceedings of the 51st Annual ACM SIGACT Symposium on Theory
  of Computing}, pages 298--309, 2019.

\bibitem[MASN16]{MASN16}
Kentaro Minami, HItomi Arai, Issei Sato, and Hiroshi Nakagawa.
\newblock Differential privacy without sensitivity.
\newblock In {\em Advances in Neural Information Processing Systems}, pages
  956--964, 2016.

\bibitem[MBST21]{MBST21}
Paul Mangold, Aur{\'e}lien Bellet, Joseph Salmon, and Marc Tommasi.
\newblock Differentially private coordinate descent for composite empirical
  risk minimization.
\newblock {\em arXiv preprint arXiv:2110.11688}, 2021.

\bibitem[Mir17]{Mir17}
Ilya Mironov.
\newblock R{\'e}nyi differential privacy.
\newblock In {\em 2017 IEEE 30th Computer Security Foundations Symposium
  (CSF)}, pages 263--275. IEEE, 2017.

\bibitem[MMW{\etalchar{+}}21]{MMW+19}
Wenlong Mou, Yi-An Ma, Martin~J Wainwright, Peter~L Bartlett, and Michael~I
  Jordan.
\newblock High-order langevin diffusion yields an accelerated mcmc algorithm.
\newblock {\em J. Mach. Learn. Res.}, 22:42--1, 2021.

\bibitem[MT07]{MT07}
Frank McSherry and Kunal Talwar.
\newblock Mechanism design via differential privacy.
\newblock In {\em 48th Annual IEEE Symposium on Foundations of Computer Science
  (FOCS'07)}, pages 94--103. IEEE, 2007.

\bibitem[MV21]{MV21}
Oren Mangoubi and Nisheeth~K Vishnoi.
\newblock Sampling from log-concave distributions with infinity-distance
  guarantees and applications to differentially private optimization.
\newblock {\em arXiv preprint arXiv:2111.04089}, 2021.

\bibitem[OV00]{OV00}
Felix Otto and C{\'e}dric Villani.
\newblock Generalization of an inequality by talagrand and links with the
  logarithmic sobolev inequality.
\newblock {\em Journal of Functional Analysis}, 173(2):361--400, 2000.

\bibitem[RBHT12]{rbht09}
Benjamin~IP Rubinstein, Peter~L Bartlett, Ling Huang, and Nina Taft.
\newblock Learning in a large function space: Privacy-preserving mechanisms for
  svm learning.
\newblock {\em Journal of Privacy and Confidentiality}, 4(1):65--100, 2012.

\bibitem[R{\'e}n61]{Ren61}
Alfr{\'e}d R{\'e}nyi.
\newblock On measures of entropy and information.
\newblock In {\em Proceedings of the Fourth Berkeley Symposium on Mathematical
  Statistics and Probability, Volume 1: Contributions to the Theory of
  Statistics}, pages 547--561. University of California Press, 1961.

\bibitem[RS16]{RS16}
Sofya Raskhodnikova and Adam Smith.
\newblock Lipschitz extensions for node-private graph statistics and the
  generalized exponential mechanism.
\newblock In {\em 2016 IEEE 57th Annual Symposium on Foundations of Computer
  Science (FOCS)}, pages 495--504. IEEE, 2016.

\bibitem[SL19]{shen2019randomized}
Ruoqi Shen and Yin~Tat Lee.
\newblock The randomized midpoint method for log-concave sampling.
\newblock In {\em Proceedings of the 33rd International Conference on Neural
  Information Processing Systems}, pages 2100--2111, 2019.

\bibitem[SSSSS09]{SSSSS09}
Shai Shalev-Shwartz, Ohad Shamir, Nathan Srebro, and Karthik Sridharan.
\newblock Stochastic convex optimization.
\newblock In {\em COLT}, volume~2, page~5, 2009.

\bibitem[SSTT21]{sstt21}
Shuang Song, Thomas Steinke, Om~Thakkar, and Abhradeep Thakurta.
\newblock Evading the curse of dimensionality in unconstrained private glms.
\newblock In {\em International Conference on Artificial Intelligence and
  Statistics}, pages 2638--2646. PMLR, 2021.

\bibitem[TS13]{TS13}
Abhradeep~Guha Thakurta and Adam Smith.
\newblock Differentially private feature selection via stability arguments, and
  the robustness of the lasso.
\newblock In {\em Conference on Learning Theory}, pages 819--850. PMLR, 2013.

\bibitem[VEH14]{EH14}
Tim Van~Erven and Peter Harremos.
\newblock R{\'e}nyi divergence and kullback-leibler divergence.
\newblock {\em IEEE Transactions on Information Theory}, 60(7):3797--3820,
  2014.

\bibitem[Wan18]{Wang18}
Yu-Xiang Wang.
\newblock Revisiting differentially private linear regression: optimal and
  adaptive prediction \& estimation in unbounded domain.
\newblock {\em arXiv preprint arXiv:1803.02596}, 2018.

\bibitem[WM10]{WM10}
Oliver Williams and Frank McSherry.
\newblock Probabilistic inference and differential privacy.
\newblock {\em Advances in Neural Information Processing Systems},
  23:2451--2459, 2010.

\bibitem[WZ10]{WZ10}
Larry Wasserman and Shuheng Zhou.
\newblock A statistical framework for differential privacy.
\newblock {\em Journal of the American Statistical Association},
  105(489):375--389, 2010.

\bibitem[ZP19]{ZP19}
Tianqing Zhu and S~Yu Philip.
\newblock Applying differential privacy mechanism in artificial intelligence.
\newblock In {\em 2019 IEEE 39th International Conference on Distributed
  Computing Systems (ICDCS)}, pages 1601--1609. IEEE, 2019.

\bibitem[ZZMW17]{zzmw17}
Jiaqi Zhang, Kai Zheng, Wenlong Mou, and Liwei Wang.
\newblock Efficient private erm for smooth objectives.
\newblock In {\em IJCAI}, 2017.

\end{thebibliography}

% \newpage
% \appendix
% \section{Appendix}
% \input{other_norms}
% \input{localization}
% \newpage
% \input{log-sobolev}

\end{document}